\documentclass[11pt]{article}
\usepackage[short,nocomma]{optidef}
\usepackage{algorithm}
\usepackage{algpseudocode}
\usepackage{mathrsfs}
\usepackage{amsmath,amssymb,amsbsy,amsgen,amsopn}
\usepackage{xspace}
\usepackage{url}
\usepackage{xspace}
\usepackage{enumerate}
\usepackage{enumitem}
\usepackage{xcolor}
\usepackage{array}
\usepackage{mathdots}
\usepackage{amsthm,amssymb,amsmath}
\usepackage{mathtools}

\newtheorem{theorem}{Theorem}
\theoremstyle{definition}
\newtheorem{definition}{Definition}
\newtheorem{lemma}[theorem]{Lemma}
\newtheorem{claim}[theorem]{Claim}
\newtheorem{corollary}[theorem]{Corollary}
\newtheorem*{theorem*}{\bf Informal Theorem}
\newtheorem{remark}[theorem]{Remark}

\newtheorem{fact}{Fact}

\newcommand{\Zset}{\mathbb{Z}}
\newcommand{\Rset}{\mathbb{R}}

\renewcommand{\P}{\mathcal{P}}
\newcommand{\R}{\mathcal{R}}
\newcommand{\F}{\mathcal{F}}
\newcommand{\mC}{\mathcal{C}}

\newcommand{\bC}{\bar{C}}
\newcommand{\hC}{\hat{C}}

\newcommand{\hy}{\hat{y}}
\newcommand{\ty}{\tilde{y}}
\newcommand{\tx}{\tilde{x}}
\newcommand{\hx}{\hat{x}}
\newcommand{\io}{i_1}
\renewcommand{\it}{i_2}
\newcommand{\ctr}{\textsf{ctr}}

\newcommand{\M}{\mathcal{M}}
\newcommand{\I}{\mathcal{I}}
\newcommand{\N}{\mathcal{N}}
\newcommand{\Q}{\mathcal{Q}}

\newcommand{\br}{\textbf{r}}
\newcommand{\f}{\textbf{f}}
\renewcommand{\a}{\textbf{a}}
\newcommand{\cI}{\mathscr{I}}
\newcommand{\cJ}{\mathscr{J}}

\newcommand{\la}{\lambda}

\newcommand{\pmatmed}{\textsf{PMatMed}}
\newcommand{\upmatmed}{\textsf{UniPMatMed}}

\newcommand{\filter}{\textsf{Filter}}

\newcommand\tab[1][1cm]{\hspace*{#1}}

\newcommand{\opt}{{\sc OPT}}

\newcommand{\chandra}[1]{{\textcolor{magenta}{Chandra: #1}}}

\usepackage{fullpage}
\usepackage{graphicx}

\usepackage[pdfencoding=auto]{hyperref}
    \hypersetup{
	colorlinks,
	linkcolor={red!50!black},
	citecolor={blue!50!black},
	urlcolor={blue!80!black}
}

\usepackage{cleveref}

\crefname{fact}{fact}{facts}
\Crefname{fact}{Fact}{Facts}
\crefname{lemma}{lemma}{lemmas}
\Crefname{lemma}{Lemma}{Lemmas}

\graphicspath{ {./images/} }

\title{Revisiting Priority $k$-Center: Fairness and Outliers}
\author{Tanvi Bajpai\thanks{Dept.\ of Computer Science, Univ.\ of Illinois, Urbana-Champaign, Urbana,
  IL 61801. {\tt tbajpai2@illinois.edu}. Partly supported by NSF grant CCF-1910149.}
  \and
  Chandra Chekuri\thanks{Dept.\ of Computer Science, Univ.\ of Illinois, Urbana-Champaign, Urbana,
  IL 61801. {\tt chekuri@illinois.edu}. Partly supported by NSF grant CCF-1910149.}
}

\title{Bicriteria Approximation Algorithms for Priority Matroid Median} 

\begin{document}

\maketitle

\begin{abstract}
  Fairness considerations have motivated new clustering problems and
  algorithms in recent years. In this paper we consider the Priority
  Matroid Median problem which generalizes the Priority $k$-Median
  problem that has recently been studied. The input consists of a set
  of facilities $\F$ and a set of clients $\mC$ that lie in a metric
  space $(\F \cup \mC,d)$, and a matroid $\M=(\F,\mathcal{I})$ over the
  facilities. In addition, each client $j$ has a specified radius
  $r_j \ge 0$ and each facility $i \in \F$ has an opening cost $f_i > 0$.
  The goal is to choose a subset $S \subseteq \F$ of facilities to
  minimize $\sum_{i \in \F} f_i + \sum_{j \in \mC} d(j,S)$ subject
  to two constraints: (i) $S$ is an independent set in $\M$ (that is
  $S \in \mathcal{I}$) and (ii) for each client $j$, its distance to
  an open facility is at most $r_j$ (that is, $d(j,S) \le r_j$).  For
  this problem we describe the first bicriteria $(c_1,c_2)$
  approximations for fixed constants $c_1,c_2$: the radius constraints
  of the clients are violated by at most a factor of $c_1$ and the
  objective cost is at most $c_2$ times the optimum cost. We also
  improve the previously known bicriteria approximation for the uniform radius setting ($r_j := L$ $\forall j \in \mC$).
\end{abstract}



\section{Introduction}

Clustering and facility-location problems are widely studied in areas
such as machine learning, operations research, and algorithm design.
Among these, center-based clustering problems in metric spaces form a
central topic and will be our focus.  The input for these problems
is a set of clients $\mC$ and a set of facilities $\F$ from a metric space
$(\F \cup \mC,d)$.  The goal is to select a subset of facilities
$S \subseteq \F$ to open, subject to various constraints, so as to
minimize an objective that depends on the distances of the clients to
the chosen centers; we use $d(j,S)$ to denote the quantity
$\min_{i \in S} d(j,i)$ which is the distance from $j$ to $S$. Typical
objectives are of the form $(\sum_{j \in \mC} d(j, S)^p)^{1/p}$ for some
parameter $p$ (the $\ell_p$ norm of the distances).  When the
constraint on facilities is that at most $k$ can be
chosen (that is, $|S| \le k$), we obtain several standard and
well-studied problems such as $k$-Center ($p = \infty$), $k$-Median
($p = 1$), and $k$-Means ($p = 2$) problems. These problems are
extensively studied from many perspectives \cite{hochbaum1985kcenter,
  plesnik1987pcenter, charikar02constant,
  ahmadian2017better,arthur2007means,kanungo2002means,jain2001median}.
These are also well-studied in the geometric setting when
$\F$ is the continuous space $\mathbb{R}^\ell$ for some finite dimension
$\ell$. In this paper we restrict our attention to the discrete
setting, and in particular, to the median objective $(p=1)$.

The Matroid Median problem is a generalization of the $k$-Median
clustering problem. Here, the cardinality constraint $k$ on $S$ is
replaced by specifying a matroid $\M=(\F,\I)$ on the facility set $\F$
and requiring that $S \in \I$ (we refer a reader unfamiliar with matroids to
Section~\ref{sec:prelim} formal definitions and details). The $k$-Median clustering problem can be written as an instance
of Matroid Median where $\M$ is the uniform matroid of rank $k$.
The Matroid Median problem was first introduced by
Krishnaswamy et al.\ \cite{krishnaswamy2011matroid} as a generalization
of $k$-Median and Red-Blue Median
\cite{hajiaghayi2010budgeted}.

Motivated by the versatility of the Matroid Median problem, and several
other considerations that we will discuss shortly, in this paper we
study the Priority Matroid Median problem ($\pmatmed$). Formally, in $\pmatmed$ we are given a set of clients $\mC$ and
facilities $\F$ from a metric space $(\F \cup \mC,d)$ where each
facility $i \in \F$ has a facility opening cost $f_i$, and each client
$j \in \mC$ has a radius value $r_j$. We are also given a matroid $\M =
(\F, \I)$ over the facilities. The goal is to select a subset of
facilities $S$ that is an independent set of the matroid $\M$ where
the objective $\sum_{j \in \mC} d(j, S) + \sum_{i \in S} f_i$ (i.e. the
\emph{cost} induced by selected facilities) is minimized, and the
\emph{radius} constraint $d(j,S) \leq r_j$ is satisfied for all
clients $j \in \mC$.

Most of the center-based clustering problems are NP-Hard even in very
restricted settings. We focus on polynomial-time approximation
algorithms which have an extensive history in center-based
clustering. Moreover, due to the nature of the constraints in
$\pmatmed$, we can only obtain bicriteria approximation guarantees
that violate both the objective and the radius constraints. An
$(\alpha,\beta)$-approximation algorithm for $\pmatmed$ is a
polynomial-time algorithm that either correctly states that no
feasible solution is possible or outputs a set of facilities
$S \in \I$ (hence satisfies the matroid constraint) such that
(i) $d(j,S) \le \alpha r_j$ for all clients $j \in \mC$ and
(ii) the cost objective value of $S$ is at most $\beta \cdot \opt$ where
$\opt$ is the cost of an optimum solution.

\subsection{Motivation, Applications to Fair Clustering, and Related Work}
Our study of $\pmatmed$ is motivated, at a high-level, by two
considerations. First, there has been past work that combines the
$k$-Median objective with that of the $k$-Center objective.  Alamdari
and Shmoys \cite{alamdari2017kcenter} considered the $k$-Median
problem with the additional constraint that each client is served
within a given uniform radius $L$ and obtained a
$(4,8)$-approximation. Their work is partially motivated by the
ordered median problem \cite{nickel2006location,aouad2019ordered,byrka18ordered}. Kamiyama \cite{kamiyama20distance}
studied a generalization of this uniform radius requirement on clients
to the setting of Matroid Median and derived a $(11,16)$-approximation
algorithm. Note that this is a special case of $\pmatmed$ where
$r_j = L$ for each $j$. We call this the $\upmatmed$ problem.

Another motivation for studying $\pmatmed$ is the recent interest
in \emph{fair clustering} in the broader context of algorithmic
fairness. The goal is to capture and address social concerns in
applications that rely on clustering procedures and algorithms.
Various notions of fair clustering have been proposed. 
Chierchetti et
al. \cite{chierichetti2017fair} formulated the Fair $k$-Center problem: clients belong
to one or more groups based on various attributes.  The objective is to return a clustering of points
where each chosen center services a representative number of clients
from \emph{every} group. This notion has since been classified as one that
seeks to achieve \emph{group} fairness. Several
other group fair clustering problems have since been introduced and
studied
\cite{bandyapadhyay2019colorful,kleindessner2019fair,abbasi2021fair,ghadiri2021socially}.
Subsequently, clustering formulations that aimed to encapsulate
\emph{individual} fairness were explored which seek to ensure that each individual
is treated fairly. One such formulation was introduced by Jung et
al. \cite{jung2019center}. This formulation is related to the
well-studied $k$-Center clustering and is the following. Given $n$
points in a metric space representing users, and an integer $k$, find
a set of $k$ centers $S$ such that $d(j,S)$ is at most $r_j$ where $r_j$ denotes the
smallest radius around $j$ that contains $n/k$ points. Such a
clustering is fair to individual users since no user will be forced to
travel outside their \emph{neighborhood}.  Jung et
al. \cite{jung2019center} showed that the problem is NP-Hard and
described a simple greedy algorithm that finds $k$ centers $S$ such
that $d(j,S) \le 2r_j$ for all $j$. Jung et al.'s model can
be related to an earlier model of Plesn\'ik who considered the Weighted
$k$-Center problem \cite{plesnik1987pcenter}. In Plesn\'ik's version,
each user $j$ specifies an \emph{arbitrary} radius $r_j > 0$ and the goal is
to find $k$ centers $S$ to serve each user within their radius
requirement. Plesn\'ik showed that a simple variant of a well-known
algorithm for $k$-Center due to Hochbaum and Shmoys \cite{hochbaum1985kcenter} yields a
$2$-approximation. Plesn\'ik's problem
has been relabeled as the Priority $k$-Center problem in recent work
\cite{bajpai2021revisiting}.

\smallskip
\noindent {\bf Priority clustering:} The model of Jung et al.\ motivated
several variations and generalizations of the Priority $k$-Center problem.
Bajpai et al.\ \cite{bajpai2021revisiting} defined, and provided constant factor approximations, for Priority 
$k$-Supplier (where facilities and clients are considered to be
disjoint sets), as well as Priority Matroid and Knapsack Center, where
facilities are subject to matroid and knapsack constraints,
respectively. Mahabadi and Vakilian \cite{mahabadi2020individual} explored
and developed approximation algorithms for Priority $k$-Median and
Priority $k$-Means problems; their motivation was to combine the individual
fairness requirements in terms of radii proposed by Jung et al., with the
traditional objectives of clustering. They obtained bicriteria approximation
algorithms via local-search. The approximation bounds were later improved via
LP-based techniques. Chakrabarty and Negahbani \cite{chakrabarty2021better}
obtained an $(8,8)$-approximation
for Priority $k$-Median and a $(8,16)$-approximation for Priority $k$-Means.
Vakilian and Yalcner \cite{vakilian2021individually} further improved these
results via a nice black box reduction of Priority $k$-Median to the Matroid Median problem! 
Via their reduction they obtained $(3, 7.081+\epsilon)$-approximation for
the Priority $k$-Median problem (relying on the algorithm for 
Matroid Median from \cite{krishnaswamy2018constant}). They extended the algorithmic ideas from Matroid Median to handle $\ell_p$ norm objectives and were thus able to derive
algorithms for Priority $k$-Means as well. The advantage of the
reduction to Matroid Median is the 
guarantee of $3$ on the radius dilation. This is optimal even for the
$k$-Supplier problem \cite{hochbaum1985kcenter}.

\subsection{Results and Technical Contribution}

In this paper, we define the $\pmatmed$ problem and derive the first $(c_1,c_2)$-bicriteria approximation algorithms where $c_1,c_2$ are both constants. There are different trade-offs between $c_1$ and $c_2$ that we can achieve. Since $\pmatmed$ simultaneously generalizes $k$-Supplier and Matroid Median, the best $c_1$ we can hope for is $3$, and the best $c_2$ that we can hope for is $\approx 8$, which comes from current algorithms for Matroid Median \cite{krishnaswamy2018constant,swamy16matroid}. We prove the following theorem which captures two results, one optimizing for the radius guarantee, and the other for the cost guarantee.

\begin{theorem} \label{thm:main-1}
  There is a $(21,12)$-approximation algorithm for the Priority Matroid Median Problem.
  There is also a $(36,8)$-approximation algorithm.
\end{theorem}

As we previously mentioned, \cite{vakilian2021individually}, via their black box reduction to Matroid Median achieve a $(3,\alpha)$ approximation for Priority $k$-Median where $\alpha$ is the best approximation for Matroid Median. We conjecture that there is a $(3,O(1))$-approximation for \pmatmed. This is interesting and open even for the special case with uniform radii under partition matroid constraint.

Our second set of results are for $\upmatmed$. Recall that \cite{kamiyama20distance} obtained a $(11,16)$-approximation for this problem. We prove the following theorem that strictly dominates the bound from \cite{kamiyama20distance}. In addition, we show that a tighter radius guarantee is achievable.

\begin{theorem} \label{thm:upmatmed}
  There is a $(9,8)$-approximation algorithm for the Uniform Priority Matroid Median Problem. For any fixed $\epsilon > 0$ there is a $(5+8\epsilon,4 + \frac{2}{\epsilon})$-approximation. 
\end{theorem}

\begin{remark}
  We believe that we can extend the ideas from this paper to obtain bicriteria approximation algorithms
  for Priority Matroid objectives that involve the $\ell_p$ norm of distances (Priority Matroid Median is when $\ell_p := 1$). Such an approximation algorithm would result in a radius factor dependent on $p$. \cite{vakilian2021individually} already showed that
  Matroid Median can be extended to the $p$-norm objective. 
\end{remark}
  
Now, we give a brief overview of our technical approach.  The reader
may wonder about the reduction  
of Priority $k$-Median to Matroid Median \cite{vakilian2021individually}. Can we make use of
this for $\pmatmed$? Indeed one can employ the same
reduction, however, the resulting instance is no longer an instance of
Matroid Median but an instance of Matroid Intersection Median which is
inapproximable \cite{swamy16matroid}. The reduction works in the
special case of Priority $k$-Median since the intersection of a
matroid with a cardinality constraint yields another matroid. We therefore address $\pmatmed$ directly. Our approximation algorithms
are based on a natural LP relaxation.  It is not surprising that we
need to build upon the techniques for Matroid Median since it is a
special case.  We build extensively on the LP-based $8$-approximation
for Matroid Median given by Swamy \cite{swamy16matroid} which improved
the first constant factor approximation algorithm of Krishnaswamy et
al.\ \cite{krishnaswamy2011matroid}.
Although the Matroid Median approximation has been improved to $7.081$
\cite{krishnaswamy2018constant}, the approach in
\cite{krishnaswamy2018constant} seems more difficult to adapt to $\pmatmed$.

Our main technical contribution is to handle the non-uniform radii
constraints imposed in $\pmatmed$ in the overall approach for Matroid
Median. We note that the rounding algorithms for Matroid Median are
quite complex, and involve several non-trivial stages: filtering,
finding half integral solutions via an auxiliary polytope, and finally
rounding to an integral solution via matroid intersection
\cite{krishnaswamy2011matroid,swamy16matroid,krishnaswamy2018constant}.
Kamiyama adapted the ideas in \cite{krishnaswamy2011matroid} to
$\upmatmed$ and his work involves four stages of reassigments that are difficult to follow. 
The non-uniform radii case introduces additional
complexity.  We explain the differences between the uniform radii
case and the non-uniform radii case briefly. The LP relaxation opens
fractional facilities and assigns each client $j$ to fractionally open
facilities.  In the LP for $\pmatmed$ we write a natural constraint
that $j$ cannot be assigned to any facility $i$ where $d(i,j) > r_j$.
Let $\bC_j$ denote the distance paid by $j$ in the LP solution.  The
preceding constraint ensures that $\bC_j \le r_j$.  For
$\upmatmed$, $r_j = L$ for all $j \in \mC$.  LP-based approximation
algorithms for $k$-Median use filtering and other rounding steps by
sorting clients in increasing order of $\bC_j$ values since they are
directly relevant to the objective. When one considers uniform radius
constraint, one can still effectively work with $\bC_j$ values since
we have $\bC_j \le L$ for all $j$. However, when clients have non-uniform
radii we can have the following situation; there can be clients
$j$ and $k$ such that $\bC_j \ll \bC_k$ but $r_j \gg r_k$. Thus the radius
requirements may not correspond to the fractional distances paid in the LP.

We handle the above mentioned complexity via two careful adaptations
to Matroid Median rounding. One of these changes occurs in the second stage of Matroid Median rounding, where we construct a half-integral solution using an auxiliary polytope. We must take care to ensure that the half-integral solution constructed in this stage is one that will not violate the radius requirements for clients. To do so, we create additional constraints in the auxiliary polytope. These constraints ensure the half-integral solution satisfies certain properties that are crucial to obtain a constant factor radius guarantee.

The second change occurs in the first filtering stage and plays a role not only for adapting Matroid Median to $\pmatmed$, but also for each of our other results. We first provide an abstract way to describe the filtering stage that allows us to specificy the order in which points are considered, and the distances each point can travel to be reassigned. For our first $\pmatmed$ result, the ordering and distances are based on both $\bC_j$ and $r_j$. For $\upmatmed$, we slightly alter the ordering and distances (using the above observations and some ideas from
\cite{kamiyama20distance}). Our remaining results will also involve changes to the filtering stage. This seems to indicate that filtering plays a large role in the cost and radius trade-off.

\medskip
\noindent
{\bf Organization:} In \Cref{sec:prelim}, we discuss preliminaries. In particular, we provide definitions and relevant information regarding matroids, define $\pmatmed$ and provide its LP relaxation, and discuss the generalized filtering procedure we will use in our algorithm. In \Cref{sec:pmatmed} we present our algorithm for $\pmatmed$ and show that it can be used to obtain $(21,12)$-approximate solutions for instances of $\pmatmed$. In \Cref{sec:upmatmed}, we describe how to modify our algorithm for $\pmatmed$ to obtain a $(9,8)$-approximate solution for instances of $\upmatmed$, and the remaining results. In \Cref{sec:pmatmed-2}, we describe how to acheive a $(36,8)$-approximate solutions for instances of $\pmatmed$. Finally, we discuss how to acheive a tighter bound for $\upmatmed$ in \Cref{sec:upmatmed-2}. 

\section{Preliminaries}\label{sec:prelim}

\subsection{Matroids, Matroid Intersection and Polyhedral Results}
We assume some basic knowledge about matroids, but provide a few
relevant definitions for sake of completeness; we refer the reader to
\cite{schrijver2003combinatorial} for more details. A matroid $\M = (S,
\I)$ consists of a finite ground set $S$ and a collection of
\emph{independent} sets $\I \subseteq 2^S$ that satisfy the following
axioms: (i) $\emptyset \in \I$ (non-emptiness of $\I$) (ii) $A \in
\I$ and $B \subset A$ implies $B \in \I$ (downward closure) and (iii)
$A, B \in \I$ with $|A| < |B|$ implies there is $i \in B \setminus A$
such that $A \cup \{i\} \in \I$ (exchange property).  The rank
function of a matroid, $r_{\M}: 2^S \to \Zset^{\geq 0}$ assigns to each
$X \subseteq S$ the cardinality of a maximum independent subset in $X$.
It is known that $r_{\M}$ is a monotone submodular function.
The matroid polytope for a matroid $\M$, denoted by $\P_{\M}$
is the convex hull of
the characteristic vectors of the independent sets of $\M$.
This can be characterized via the rank function:
\[ \P_\M = \{v \in \Rset^S \mid \forall X \subseteq S:~v(X)
\leq r_M(X)~\textnormal{and}~\forall e \in S:~v(e) \geq 0\}.
\]
Assuming an independence oracle\footnote{An independence oracle returns whether $A \in \I$ for a given $A \subseteq S$.} or a rank function oracle for $\M$,
one can optimize and separate over $\P_{\M}$ in polynomial time.
A \emph{partition matroid} $\M = (S, \I)$ is a special type of matroid
that is defined via a partition $S_1,S_2,\ldots,S_h$
of $S$ and non-negative integers $k_1,\ldots,k_h$.
A set $X \subseteq S$ is independent, that is $X \in \I$, iff
$|X \cap S_i| \le k_i$ for $1 \le i \le h$. A simple partition matroid
is one in which $k_i = 1$ for each $i$.

Given two matroids $\M=(S,\I_1)$ and $\N = (S,\I_2)$, on the same ground set, their
intersection is defined as $\M \cap \N := (S,\I_1 \cap \I_2)$ consisting of sets
that are independent in both $\M$ and $\N$. Computing a maximum weight
independent set in the intersection can be done efficiently.
The convex hull of the characteristic vectors of the independent
sets of $\M \cap \N$, denoted by $\P_{\M,\N}$, is simply the intersection
of $\P_{\M}$ and $\P_{\N}$! That is
\begin{align*}
    \P_{M,N} = \{v \in \Rset_+^S \mid \forall X \subseteq S:~v(X) \leq r_{M}(X), v(X) \le r_{\N}(X)\}. 
\end{align*}
Thus, one can optimize over $\P_{\M,\N}$ if one has independence or
rank oracles for $\M$ and $\N$. We will need these results later in
the paper. See \cite{schrijver2003combinatorial} for these classical results.

The input matroid $\M$ for Priority Matroid Median has ground set $\F$ i.e.\ the set of facilities. Thus, an integer point of the polytope
$v^* \in \P_\M$ will represent a subset of facilities that is an
independent set of the matroid $\M$.

\subsection{Priority Matroid Median}

We provide below a more general definition of Priority Matroid Median that includes a notion of client \emph{demands}.

\begin{definition}[$\pmatmed$]
The input is a set of facilities $\F$ and clients $\mC$ from a metric
space $(\F \cup \mC, d)$. Each $i \in \F$ has an opening cost
$f_i\ge 0$. Each client $j \in \mC$ has a radius value, $r_j \ge 0$ and a demand
value $a_j \ge 0$. We are also given a matroid $\M = (\F, \I)$. The goal is to
choose a set $S \in \I$ to minimize $\sum_{i \in S} f_i + \sum_{j \in
  \mC} a_j d(j, S))$ with the constraint that $d(j, S) \leq r_j$ for each $j \in \mC$.
\end{definition}

A $\pmatmed$ instance $\cI$ is the tuple $(\F,\mC, d,\f,\br,\a,\M)$, where $\f \in \Rset^\F$ and $\br,\a \in \Rset^\mC$. 

\subsection{LP relaxation for \pmatmed}\label{sec:prelim-lp}

Our algorithm is based on an LP relaxation for a $\pmatmed$ instance
$\cI = (\F,\mC, d,\f,\br,\a,\M)$ that we describe next. We use $i$ to
index facilities in $\F$, $j$ to index clients in $\mC$. Recall that
$r_{\M}$ denotes the rank function of the matroid $\M$. The $y_i$
variables denote the fractional amount a facility $i$ is open, while
the $x_{ij}$ variables indicate the fractional amount a client $j$ is
assigned to facility $i$.

\begin{mini!}
  {}{\displaystyle\sum\limits_{i \in \F} f_i y_i + \displaystyle\sum\limits_{j} \displaystyle\sum\limits_{i} a_j  d(i,j) x_{ij}}{}{}
  \addConstraint{\displaystyle\sum\limits_{i \in \F} x_{ij}}{\geq 1}{\quad\forall j \in \mC \label{cover}}
    \addConstraint{x_{ij}}{\leq y_i}{\quad\forall i \in \F, j \in \mC \label{open}}
  \addConstraint{x_{ij}}{= 0}{\quad\forall i \in \F, j \in \mC: d(i,j) > r_j \label{priority}}
    \addConstraint{y}{\in \P_\M}{ \label{matroid}}
   \addConstraint{x_{ij},y_i}{\ge 0}{\quad \forall i \in \F, j \in C \label{relax}}
\end{mini!}

Constraint \ref{cover} states that each client must be fully assigned
to facilities, and constraint \ref{open} ensures that these facilities
have indeed been opened enough to service clients. For integral $y$,
constraint \ref{matroid} mandates that the facilities come from an
independent set of the matroid $\M$. Finally, constraint
\ref{priority} ensures that no client is assigned to a center that is
farther than its radius value. 

We make a few basic observations about the LP relaxation.
We assume that it is feasible for otherwise
the algorithm can terminate reporting that there is no feasible integral solution. Indeed, the LP is solvable in polynomial time via the rank oracle for $\M$. 
First, some notation. For $X \subseteq \F$, we let
$y(X)$ denote $\sum_{i \in X} y_i$. For client $j$ and radius parameter
$R$ we let $B(j,R)$ denote the set $\{i \in \F \mid d(i,j) \leq R\}$
of facilities within $R$ of $j$.
Constraints \ref{cover} and \ref{priority} ensure the following simple fact.
\begin{fact} 
\label{fact:y1}
    Let $(x,y)$ be a feasible solution to the $\pmatmed$ LP. Then $y(B(j, r_j)) \geq 1$  holds $\forall j \in \mC$. 
\end{fact}

Let $COST(x,y)$ denote the cost of the LP using solution $(x,y)$.
Going forward, we will assume that we are working with an optimum
fractional solution to the LP relaxation for the given instance.

\begin{remark}
  \label{optxfromy}
  We say that $y$ is feasible if $y \in \P_{\M}$ and $y(B(j,r_j))
  \ge 1$ for all $j \in C$.  Given feasible $y$, a corresponding $x$
  satisfying the constraints and minimizing $COST(x,y)$ is determined
  by solving a min-cost assignment problem for each client $j \in C$
  separately.
\end{remark}

\subsection{Filtering}\label{sec:prelim-filter}
Filtering is a standard step in several approximation algorithms for
clustering and facility location wherein one identifies a subset of
well-separated and representative clients. Each client is assigned to
a chosen representative. In priority median problems there are two
criteria that dictate the filtering process. One is the radius upper
bound $r_j$ for the client $j$. The other is the LP distance $\bC_j =
\sum_{i} d(i,j)x(i,j)$ paid by the client which is part of the
objective. Balancing these two criteria is important. To facilitate
different scenarios later we develop a slightly abstract filtering
process.
Building on a procedure introduced in
\cite{hochbaum1985kcenter,plesnik1987pcenter}, $\filter$ takes in the
metric and demands from a $\pmatmed$ instance $\cI = (\F,\mC,
d,\f,\br,\a,\M)$, as well as functions $\phi,\la: \mC \to \Rset_+$ that
satisfy the following condition.

\begin{definition}[compatibility] \label{def:comp}
Functions $\phi,\la: \mC \to \Rset_+$ are compatible if for any ordering of clients $j_1,j_2,\hdots,j_n$ where $\phi(j_1) \leq \phi(j_2) \leq \hdots \leq \phi(j_n)$, it is the case that $\la(j_1) \leq \la(j_2) \leq \hdots \leq \la(j_n)$. 
\end{definition}

\remark{This condition trivially holds when $\phi$ and $\la$ are
  identical. The filtering stages of many clustering approximation
  algorithms \cite{bajpai2021revisiting,jung2019center,chakrabarty2021better}
  utilize equal $\phi$ and $\la$ functions. We use both identical and non-identical settings for $\phi$ and $\la$ in this paper.}

The function $\phi$ encodes an ordering of clients, while $\la$
represents a client's coverage distance. $\filter$ chooses
cluster centers in order of increasing $\phi$ values, and then
``covers'' any remaining client $k$ that is within distance $2\cdot \la(k)$ from
the newly added center $j$. The demand from the covered points is transferred
to the center that first covered them. The new demand variables $a'$
represent the aggregated demand for the chosen centers. $\filter$ returns the set of
cluster centers, the clusters assigned to each cluster center, and new
demand assignments for all clients. 

\begin{algorithm}[ht]

	\caption{$\filter$}
	\label{alg:filter}
	\begin{algorithmic}[1]
		\Require Metric $(\F \cup \mC,d)$, demands $\a$, compatible functions $\phi,\la: \mC \to \Rset_{>0}$
		\State $U \leftarrow \mC$ \Comment{The set of uncovered clients} \label{ln:1}
		\State $C \leftarrow \emptyset$ \Comment{The set of cluster centers}
		\State $\forall j \in \mC$ set $a'_j := 0$ \Comment{Initialize new demand variables}
		\While{ $U \neq \emptyset$}
	        \State $j \leftarrow \arg\min_{j\in U} \phi(j)$ 
	        \State $C \leftarrow C \cup \{j\}$
	        \State $D(j) \leftarrow \{k \in U: d(j,k) \leq 2\cdot \la(k) \}$\Comment{Note: $D(j)$ includes $j$ itself} \label{ln:chld}
	            \State $a'_j = \sum_{k \in D(j)} a_k$ \Comment{Accumulate all demands of $D(j)$ to $j$ }
	        \State $U \leftarrow U \backslash D(j)$ \label{ln:remove-from-U}
		\EndWhile
		\State {\bf Return} cluster centers $C$, $\{D(j) : j \in C\}$, updated demands $\a' \in \Rset^\mC$
	\end{algorithmic}
\end{algorithm}
The resulting cluster centers $C \subseteq \mC$, and the sets of
clients relocated to each cluster center $\{D(j) \mid j \in C \}$ form
a partition of the client set $\mC$.  When the given $\phi$ and $\la$
are compatible, the returned clusters satisfy certain desirable
properties, described in the following facts which are relatively
easy to see, and standard in the literature. For this reason we omit formal proofs.

\begin{fact}\label{fact:filter}
    The following statements hold for the output of $\filter$: (a) $\forall j,j' \in C, d(j,j') > 2\max\{\la(j),\la(j')\}$. (b) $\{B(j, \la(j)) \mid j \in C\}$ are mutually disjoint. (c) $\{D(j) \mid j \in C\}$ partitions $\mC$. (d) $\forall j \in C, \forall k \in D(j), \phi(j) \leq \phi(k)$ and $\la(j) \leq \la(k)$. (e) $\forall j \in C, \forall k \in D(j), d(j, k) \leq 2\cdot \la(k)$
\end{fact}




\noindent
{\bf Choosing $\phi$ and $\la$:} As we remarked, the two criteria that
influence the filtering process are $r_j$ and $\bC_j$. For the
algorithm in Section~\ref{sec:pmatmed} we choose $\phi(j) = \la(j) = \min\{r_j,
2\bC_j\}$.  There are other valid settings of compatible $\phi$ and
$\la$ that can be used in the filtering stage. Different settings of
$\phi$ and $\la$ will result in different approximation factors for
cost and radius. This is showcased via the results of Sections 4 through 6, since each require altering the settings of $\phi$ and $\la$.

\section{A $(21,12)$-approximation for Priority Matroid Median}\label{sec:pmatmed}

Our algorithm will follow the overall structure of the LP-based
procedure used for approximating Matroid Median from
\cite{swamy16matroid}, but will contain a few key alterations that
allow us to be mindful of the radius objective of $\pmatmed$. Stage 1
of our algorithm involves filtering the client set to construct an
updated instance $\cI'$ using the cluster centers and updated
demands. We will show that a solution to $\cI'$ can be converted to a
solution for $\cI$ while only incurring a small increase to the cost
and radius. The focus then shifts to constructing a solution for $\cI'$. In Stage 2, we
obtain a half-integral solution for the LP-relaxation for $\cI'$ by
working with an auxiliary polytope. In Stage 3, this half-integral solution is converted to an
integral solution for $\cI'$. This is done via a reduction to
matroid intersection. Finally, we will show that this solution yields a
$(21,12)$-approximation for the original instance $\cI$. \Cref{alg:app} is given as a summary of the various stages of our algorithm. 

\begin{algorithm} 

\caption{Overview of bi-criteria approximation algorithm for $\pmatmed$. \\ \textbf{Input:} $\pmatmed$ instance $\cI = (\F,\mC, d,\f,\br,\a,\M)$. 
\\ \textbf{Output:} $(\alpha,\beta)$-approximate solution for $\cI$.}
\label{alg:app}
\begin{algorithmic}[1]
\setcounter{ALG@line}{-1}

\State Solve $LP$ for $\cI$ and let $(x,y)$ denote the optimal fractional solution. Use $(x,y)$ and radius values $\br$ to help set $\phi$ and $\lambda$. \label{ln:lp}  
\State \textbf{Stage 1 -} Run $\filter((\F \cup \mC, d), \a, \phi,\lambda)$ which returns cluster centers $C$, and updated client demands $\a'$. Create an updated instance $\cI' = (\F,C, d,\f,\br,\a',\M)$ (\Cref{sec:pmatmed-filter}). \label{ln:filter} 
\State \textbf{Stage 2 -} Construct a half-integral solution $(\hx,\hy)$ for $\cI'$ by setting up a polytope $\Q$ with half-integral extreme points (\Cref{sec:pmatmed-half}). \label{ln:half} 
\State \textbf{Stage 3 -} Convert the half-integral solution to an integral solution $(\tx,\ty)$ for $\cI'$ by setting up an instance of matroid intersection between the input matroid $\M$, and a partition matroid $\N$ constructed with respect to the half-integral solution (\Cref{sec:pmatmed-int}). \label{ln:int}
\State Convert the integral solution for $\cI'$ to one for $\cI$ (\Cref{lm:cost-filter}). \label{ln:convert}
\end{algorithmic}
\end{algorithm}

\subsection{Stage 1: Filtering Clients}\label{sec:pmatmed-filter}
In this stage, we create a new instance of $\pmatmed$ from the initial one
by using the $\filter$ process described in \Cref{sec:prelim-filter}.
Recall that $\filter$ will return a set of cluster
centers $C \subseteq \mC$, and collections of clients that are
relocated to each cluster center $\{D(j) \mid j \in C \}$. $\filter$
also returns a set of updated demands for all clients, $\a'$.
Now, using $C$ and $\a'$, we construct a new instance of
$\pmatmed$ $\cI' = (\F,C, d,\f,\br,\a',\M)$. Here, we overload notation and take $\br$
and $\a'$ to denote the vector of radius values and demands,
respectively, for cluster centers (i.e. $\br, \a' \in \Rset^C$). Notice
that we do not lose any information by restricting $\a'$ to $C$, since
the updated demands for relocated points are set to $0$. Furthermore,
we will reconcile the radius objective for relocated points in the
final solution at the end of the section.

The solution $(x,y)$ for instance $\cI$, when restricted to $C$, will
still be a feasible solution for the LP for $\cI'$, since the new LP
is made up of a subset of constraints from the original LP. For
updated instance $\cI'$, we denote the cost of the
LP solution $(x,y)$ by $COST'(x,y)$.
\[ COST'(x,y) = \sum_{i \in \F} f_iy_i + \sum_{j \in C} a'_j \sum_{i
    \in \F} d(i,j) x_{ij} = \sum_{i \in \F} f_iy_i + \sum_{j \in C}
  a'_j \bC_j \]

The next lemma shows that an integral solution to $\cI'$
can be translated to an integer solution for $\cI$ by incurring a
small additive increase to the cost objective. In subsequent sections
we will address how the translated solution also ensures that
all clients are served within a constant factor of their radius constraint.

\begin{lemma} \label{lm:cost-filter} The following is true of
  $\cI'$: (a) $COST'(x,y) \leq 2\cdot COST(x,y)$. (b) Any integer
  solution $(x',y')$ for $\cI'$ can be converted to an integer
  solution for $\cI$ that incurs an additional cost of at most
  $4\cdot COST(x,y)$.
\end{lemma}

\begin{proof}
  We first prove that $COST'(x,y) \leq 2\cdot COST(x,y)$.  The
  fractional facility opening cost, $\sum_{i} f_i y_i$ is identical in
  both. The difference in the client connection cost is because the
  demands of clients in $\mC \setminus C$ are relocated. Consider a
  client $k \in \mC \setminus C$ that is relocated to its cluster
  center $j \in C$ (thus $k \in D(j)$).  In $COST(x,y)$ client $k$
  pays $a_k \bC_k$.  In $COST'(x,y)$, the demand of $k$ is moved to
  $j$ and pays $a_k \bC_j$.  Thus, it suffices to prove that $\bC_j
  \le 2\bC_k$.  From \Cref{fact:filter}, $\phi(j) \leq \phi(k) \leq
  2\bC_k$.  LP constraints \ref{priority} and \ref{open} of the LP for
  $\cI$ ensures that $\bC_j \leq r_j$ for all $j \in $. Hence, if
  $\bC_j > 2\bC_k$ we would have $\phi(j) = \min\{r_j,2\bC_j\} > 2\bC_k$ which
  would be a contradiction to $\phi(j) \le \phi(k)$. This shows that
  $\bC_j \le 2\bC_k$.

  Now we consider the second part. From \Cref{fact:filter},
  $d(j,k) \leq 2\la(k) \leq 2(2\cdot \bC_k)$. Suppose the cost
  of an integer solution to $\cI'$ is $\alpha$. We keep the
  same facilities for $\cI$ and account for the increase in connection cost
  when considering the original client locations. Consider a client $k \in \mC  \setminus C$
  that is relocated to center $j \in C$. If $j$ connects to $i$ in the integer solution
  for $\cI'$, $k$ can connect to $i$ in the solution to $\cI$, and its per unit connection cost
  increases by at most $d(j,k) \le 4 \bC_k$. Thus the total
  increase in the connection cost when comparing to $\alpha$ is
  upper bounded by $\sum_{j \in C} \sum_{k \in D(j)} a_k \cdot 4 \bC_k \le 4 \cdot COST(x,y)$.
\end{proof}


The following lemma follows directly from \Cref{fact:filter}. 

\begin{lemma}\label{lm:rad-filter}
  Let $k \in \mC$ be assigned to $j \in C$ after $\filter$ (i.e.
  $k \in D(j)$). Then, $d(j,k) \leq 2\la(k) \leq 2r_{k}$.
\end{lemma}

\subsection{Stage 2: Constructing Half-Integral Solution $(\hx,\hy)$}\label{sec:pmatmed-half}

In the second stage the goal is to construct a half-integral solution
to $\cI'$. This means that each cluster center/client $j \in C$ will
connect to at most two facilities. This is accomplished by
constructing a specific polytope $\Q$ with only facility variables, and
a proxy objective that also has only facility variables and arguing
about the properties of $\Q$ and the objective function. 



To describe $\Q$, we define, for each client
$j \in C$, several facility sets that will play an important role.
Let $F_j = \{i \in \F \mid d(i,j) = \min_{k \in C} d(i,k)
\}$ denote the set of facilities $i$ for which $j$ is the closest
client in $C$ (ties are broken arbitrarily). Let $F'_j =
\{ i \in F_j \mid d(i,j) \leq \la(j) \} \subseteq F_j $. 
Let $\gamma_j := \min_{i
  \notin F_j} d(i,j)$ denote the distance between client $j \in C$ and
the closest facility $i$ not included in $F_j$. In other words, $i$ in
the definition of $\gamma_j$ is the closest facility to $j$ that has
some other closest cluster center $j' \in C$ such that $j \neq
j'$. Using $\gamma_j$,  let $G_j = \{ i \in F_j \mid d(i,j) \leq
\gamma_j\}$. Finally, let $\rho_j$ be the smallest distance such that
$y(B(j,\rho_j)) \geq 1$, and $B_j := B(j,\rho_j)$.\footnote{Note that though it may be the case that $y(B(j, \rho_j)) > 1$, we can split facilities and define $B_j$ as the points of $B(j, \rho_j)$ such that $y(B_j) = 1$.} See \Cref{fjbjgj}.



We summarize some basic properties of the defined sets below.

\begin{fact}
  \label{fact:fjgj}
  The following hold for all $j \in C$:
(a) If $j' \neq j$, $F_j \cap F_{j'} = \emptyset$; (b) $F_j$ contains all the facilities $i$ such that $d(i,j) \leq \la(j)$; (c) $\gamma_j > \la(j)$; (d) $F'_j \subseteq G_j$; (e) $\rho_j \leq r_j$, (f) $\sum_{i \in F'_j} x_{ij} \geq 1/2 $ and when $\la(j) = r_j$, $\sum_{i \in F'_j} x_{ij} = 1$;
\end{fact}

\begin{proof}
  (a) follows from definition of $F_j$, (b), (c), (d) follow from \Cref{fact:filter}(b) and definitions.
  (e) follows from the LP constraint. We now prove (f). If $\la(j) = r_j$, $F'_j = \{i \mid d(i,j) \le r_j \}$,
  and by LP constraint $\sum_{i \in F'_j} x_{ij} = 1$. Otherwise $\la(j) = 2\bC_j < r_j$. Note that
  $\bC_j = \sum_{i} d(i,j)x_{ij}$. By averaging argument (Markov's inequality) we have
  $\sum_{i: d(i,j) \le 2\bC_j} x_{ij} \ge 1/2$. This gives the desired claim since
  $F'_j = \{i \mid d(i,j) \le \la(j)\}$.
\end{proof}

At this point in the algorithm, in a departure from the Matroid Median
algorithm of \cite{swamy16matroid}, we need to be mindful of two
cases. If $\rho_j \leq \gamma_j$, in order to satisfy the radius
requirements of $\pmatmed$, it is important to open one facility
within radius $\rho_j$ of $j$. If it is the case that $\rho_j > \gamma_j$,
it is not necessary to do so. To distinguish these two cases, we
partition $C$ into $C_s = \{ j \in C \mid \rho_j \leq \gamma_j \}$, and
$C_b = \{ j \in C \mid \rho_j > \gamma_j\}$.  For $j \in C_s$, it should
be clear that $B_j \subseteq G_j$. Using these sets, we define a polytope
$\Q$ with facility variables $v_i, i \in \F$ as follows. It consists of
the matroid constraints induced by $\M$ and a second set of
constraints induced by $C$ and $C_s$ as defined above. In particular, we require that all points $j$ in $C$ has at least $1/2$ value assigned cumulatively to facilities within their $F'_j$ balls. We require points of $C_s$ to have exactly $1$ assigned to facilities within $B_j$.


\begin{align*}
    \Q = \Big\{ v \in \Rset^\F_{\geq 0} ~\mid~ \forall S \subseteq \F:v(S) \leq r_\M(S),~~&\forall j \in C:~v(F'_j) \geq 1/2~\textnormal{and}~v(G_j) \leq 1,\\ &\forall j \in C_s: v(B_j) = 1\Big\}
\end{align*}


\begin{figure}
    \centering
    \includegraphics[width =\textwidth]{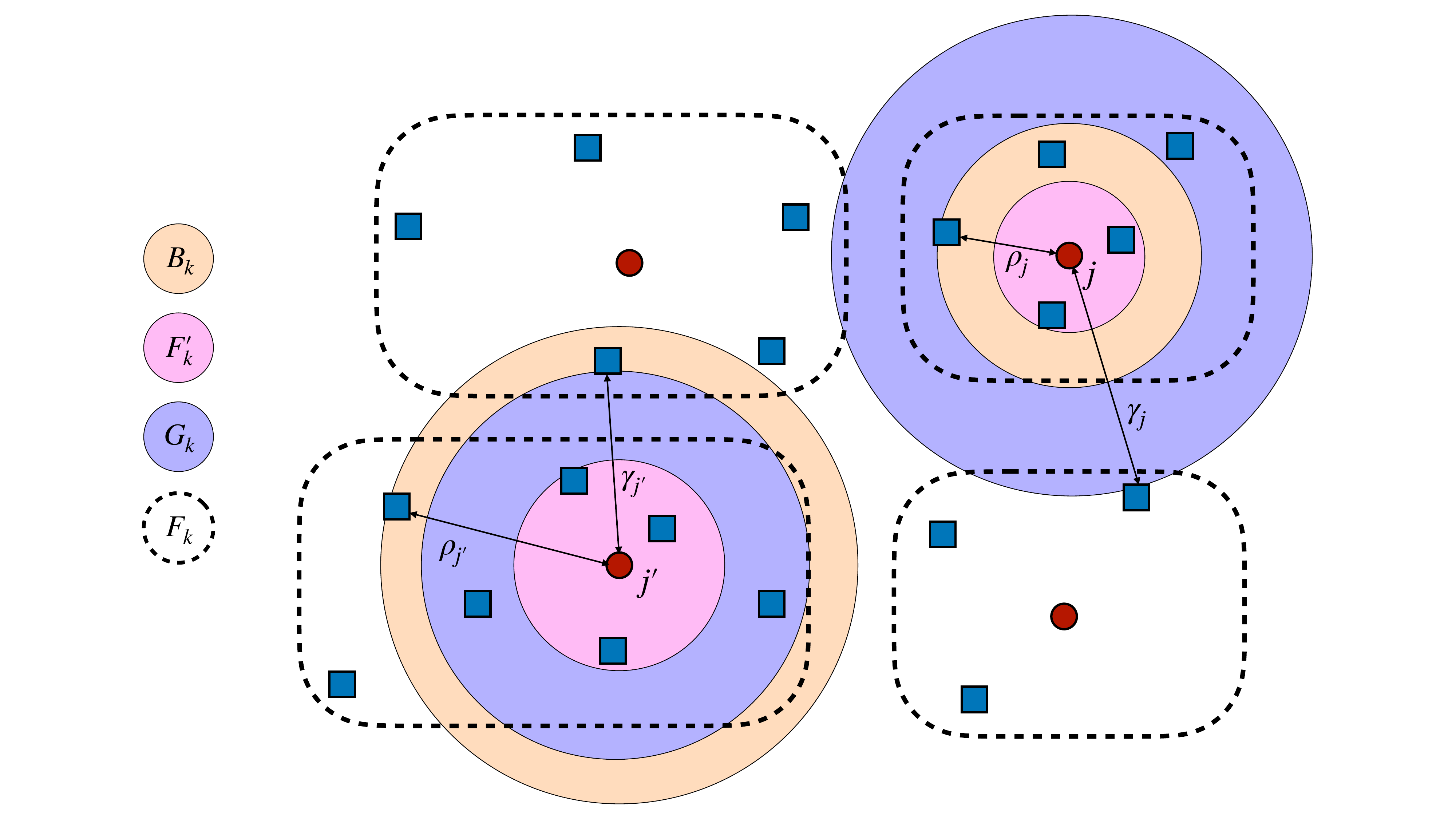}
    \caption{The $F$, $F'$, $G$, and $B$ sets for points $j \in C_s$ and $j' \in C_b$. Observe that for $j$, $\rho_j \leq \gamma_j$, hence $B_j \subseteq G_j$.}
    \label{fjbjgj}
\end{figure}

\begin{lemma}\label{lm:poly}
  The extreme points of the polytope $\Q$, if non-empty, are half-integral.
\end{lemma}

The proof of the preceding lemma is similar to those in previous works on Matroid Median \cite{krishnaswamy2011matroid,swamy16matroid}.
We give a proof for the sake of completeness since the polytope we define is slightly different
due to the separation of clients in $C$ into $C_s$ and $C_b$ in order to enforce an additional constraint.

\begin{proof}
  Suppose $\Q$ is non-empty and $v^*$ is any extreme point.
  Then $v^*$ is the unique solution of a linear system
  $Av = b$ where $A$ is a subset of the inequalities of $\Q$
  with $A$ having full row and column rank (in particular the rows of $A$ are linearly
  independent vectors). $A$ can be partitioned into $A_1$ and $A_2$
  where $A_1$ is a subset of the inequalities coming from the
  matroid $\M$ (of the form $v(S) = r_{\M}(S)$), while $A_2$ is a subset
  of the remaining inequalities. Via the submodularity of the matroid rank function,
  it is known that one can choose $A_1$ such that the rows of $A_1$
  correspond to a laminar family of subsets of $\F$ \cite{schrijver2003combinatorial}.
  We observe that the non-matroidal system of inequalities in $\Q$ correspond
  to a laminar family of sets over $\F$: (a) the sets $G_j, j \in C$ are disjoint
  and $F'_j \subseteq G_j$ for each $j$ and (b) for $j \in C_s$, we have $B_j\subseteq G_j$.
  See \Cref{fjbjgj}.

  Thus the rows of the matrix of $A$ come from two laminar
  families of sets over $\F$, and it is known that such a matrix is
  totally uniodular \cite{schrijver2003combinatorial}.  Thus $v^* =
  A^{-1}b$ where $A^{-1}$ is an integer matrix, and $b$ is
  half-integral which implies that $v^*$ is half-integral.
\end{proof}

We will now define a vector $y'$ that lies in $\Q$ which will prove
that it is non-empty. Further, we also define a linear objective
function $T(\cdot)$ over vectors in $\Q$ to serve as a proxy for the
cost.  Following the analysis for the improved bound in
\cite{swamy16matroid}, we set up $T(\cdot)$ with some slack so that
the slack can be exploited in the analysis of the next step in the
algorithm.


We define $y' \in \Rset^\F_{\geq 0}$ as follows. For all $j \in C$ and $i \in G_j$, set $y'_i = x_{ij} \le y_i$. For a facility $i \notin \cup_j G_j$ set $y'_{i} = 0$. From this definition it should be clear that $y'(G_j) \leq 1$ for all $j \in C$, since $\sum_{i \in G_j} x_{ij} \leq 1$. Also, from
\Cref{fact:fjgj}, $y'(F'_j) \ge 1/2$. For $j \in C_s$, it will be the case that $y'(B_j) = 1$
since $\sum_{i \in B_j} x_{ij} = 1$; we also know that for these points, $y'(G_j) = y'(B_j)$. 

To build up to the definition of $T$, we first state the
following lemma, which we will prove in the proof of \Cref{lm:cost-4bound}.

\begin{lemma}\label{lm:cost-gambound}
Consider some $j \in C$, and let $i$ and $j'$ be the facility and cluster used to define $\gamma_j$ (i.e. $\gamma_j = d(i,j)$) where $i \in
F_{j'}$ for some ${j'} \neq j$. For every $i' \in F'_{j'}$, $d(i', j)
\leq 3\gamma_j$.
\end{lemma}

Keeping the preceding lemma in mind, we can use as proxy for $j$'s
per-unit-demand cost a function written in terms of the facility vector $v$. When
$y'(G_j) = 1$, the cost for $j$ can be bounded by
$\sum_{i \in G_j} d(i,j) y'_i \leq \bC_j$. When
$y'(G_j) < 1$, the preceding lemma indicates that we can upper bound the
cost of the solution by $\sum_{i \in G_j} d(i,j) y'_i +
3\gamma_j(1-y'(G_j)) \leq 3 \cdot \bC_j$. Using these two bounds, we
define $T(\cdot)$ for $v \in \Q$ as follows:


\[ T(v) = \sum_{i \in \F} f_i v_i + \sum_{j \in C} a'_j \Big( 2 \sum_{i \in G_j} d(i,j) v_i + 4\gamma_j(1-v(G_j))\Big) \]


For a vector $v$ such that $v(F'_j) \geq 0.5$ and $v(G_j) \leq 1$ for all $j \in C$, the term $a'_j (2\sum_{i \in G_j} d(i,j) y'_i + 4\gamma_j(1-y'(G_j)))$ will upper bound $j$'s assignment cost with respect to $v$ via \Cref{lm:cost-gambound}. When $v(G_j) = v(B_j) = 1$ for $j \in C_s$, $j$'s assignment cost will be at most $a'_j (2\sum_{i \in B_j} d(i,j) v_i)$.
Indeed $T(v)$ is an \emph{overestimate} and we will use this in the next
step.

We find an optimum half-integral solution $\hy$ to $\Q$ with objective $T(v)$.
It follows that $T(\hy) \le T(y')$. Now, we construct a half-integral solution
$(\hx,\hy)$ from $\hy \in \Q$: For each cluster center $j \in C$,
if $\hy(G_j) = 1$, set $\sigma(j) = j$. Otherwise, set
$\sigma(j) = \arg\min_{j' \in C: j'\neq j} d(j,j')$. Now, the
\emph{primary facility} for each cluster center is the closest
facility $i \in \F$ such that $\hy_i > 0$ (this will always be located
in $F'_j$), is denoted by $\io(j)$, and thus
$\hx_{\io(j)j} = \hy_{\io(j)}$. A cluster center's \emph{secondary}
facility, denoted by $\it(j)$, is the next option of facility for $j$
to use, when it cannot be completely serviced by its primary
facility. If $\hy_{\io(j)} = 1$, then $j$ does not need a secondary
facility, since $\io(j)$ has been completely opened, and will remain
completely opened. When $\hy_{\io(j)} < 1$ and $\hy(G_j) = 1$, then
set $\it(j)$ to be the second closest partially opened facility to $j$
(where $\hy_{\it(j)} > 0$).  Otherwise, when $\hy_{\io(j)} < 1$ and
$\hy(G_j) < 1$, we now set $\it(j) = \io(\sigma(j))$ and
$\hx_{\io(j)} = \hx_{\it(j)} = 1/2$. Note that if $j \in C_s$ then
$\hy(B_j) = 1$ which implies that $j$'s primary and secondary facilities
are both in $B_j$ and $\sigma(j) = j$.
The following two claims are easy to see.

\begin{claim}\label{clm:sig-2} For all $j \in C$, $d(j,\sigma(j)) \leq 2 \gamma_j$.\end{claim}

\begin{claim}\label{claim:cscb}
For all $j \in C_s$,  $\hy(G_j) = 1$.  If $\hy(G_j) < 1$, it must be the case that $j \in C_b$.
\end{claim}

By \Cref{fact:filter}(b), each $j$ will have a unique primary facility
that is at least partially opened in $F'_j$. For points $j \in C_s$,
their secondary facility must be in $B_j$. However, for points in $j
\in C_b$, $\it(j)$ might not be in $G_{j}$ or even $F_{j}$. As per
\Cref{lm:cost-gambound}, we know that $j$ will be able to find a
partially open facility to be serviced by that is within distance $3
\gamma_{j} < 3\rho_j$. In the following lemma, we derive our bound for the cost
of $(\hx,\hy)$.

\begin{lemma}\label{lm:cost-4bound}
$COST'(\hx,\hy) \leq T(\hy) \leq T(y') \leq 4 \cdot COST'(x,y) \leq 8 \cdot COST(x,y)$.
\end{lemma}

\begin{proof}
  We first show that $T(y') \leq 4 \cdot COST'(x,y)$ (we already have $T(\hy) \leq T(y')$).
  We know that $COST'(x,y)$ can be expressed
as $\sum_i f_iy_i + \sum_j a'_j\cdot \bC_j$. For any $j \in C$,
observe that $\bC_j = \sum_{i \in G_j} d(i,j) x_{ij} + \sum_{i \notin
  G_j} d(i,j)x_{ij}$ and hence $\bC_j \geq \sum_{i \in G_j} d(i,j)
x_{ij} + \gamma_j \sum_{i \notin G_j} x_{ij}$.

\begin{align*}
     T(y') &\leq \sum_i f_iy_i + \sum_j a'_j \Big(2\sum_{i \in G_j} d(i,j)x_{ij} + 4\gamma_j\Big(1-\sum_{i \in G_j} x_{ij}\Big)\Big) \\
     &\leq \sum_i f_iy_i + 4 \sum_j a'_j \cdot \bC_j \leq 4 \cdot COST'(x,y)
\end{align*}

Next, we upper bound $COST'(\hx,\hy)$ by $T(\hy)$. It suffices to focus
on the assignment cost. Consider $j \in
C_s$. Its primary and secondary facilities are in $B_j$ and it is easy
to see that its connection cost is precisely $\sum_{i \in B_j}
d(i,j)\hx_{ij}$. Now consider $j \in C_b$. Recall that when $\hy(G_j) = 1$,
the total assignment cost of $j$ is at most $\sum_{i \in G_j} d(i,j)
\hy_i$.  When $\hy(G_j) < 1$, $j$ connects to primary facility in
$F'_j$ and a secondary facility. The second nearest facility will not be
in its $G_j$ ball, i.e. $\it(j) \notin F_j$.  Let $j' \neq j$ be
client that defines $\gamma_j$. Via \Cref{lm:cost-gambound}, we have
$d(\it(j), j) \leq 3 \gamma_j$. Assuming this, when $\hy(G_j) < 1$,
the total assignment cost of $j$ is at most $\sum_{i \in G_j} d(i,j)
\hy_i + 3\gamma_j (1-\hy(G_j))$. Based on these assignment cost
upper bounds we see that $COST'(\hx,\hy) \le T(\hy)$.

Now we prove \Cref{lm:cost-gambound}.
From \Cref{fact:filter} we have $2\max\{\la(j),\la(j')\} \leq d(j,j')$.
Via triangle inequality $d(j,j') \le d(j,i) + d(i,j') \leq 2\gamma_j$.
Thus $2\la(j') \leq 2\gamma_j$ which implies that $\la(j') \leq \gamma_j$.
Recall that
$F'_{j'}$, from its definition, is contained in a ball of radius $\la(j')$
around $j'$. Thus, for any facility $i' \in F'_{j'}$, $d(i',j') \leq \la(j')\le \gamma_j$,
Therefore, $d(i',j) \leq d(j,j') + d(j',i') \leq 3\gamma_j$. This gives
us the lemma.

Finally, using \Cref{lm:cost-filter}, we know that $COST'(x,y) \leq
2\cdot COST(x,y)$, hence $4\cdot COST'(x,y) \leq 8 \cdot COST(x,y)$.
\end{proof}

Before moving on to the final stage of the algorithm, we prove
a few lemmas that will be relevant for our analysis of the radius dilation
of the final solution. \Cref{lm:rad-rk} allows us to relate the radius
of cluster center $j$ to that of a client $k$ in the original instance
that was relocated to $j$. We need such a lemma because even though we
know that $\phi(j) = \min\{r_j, 2\bC_j \}$ and $\phi(j) \leq \phi(k)$
for all $k \in D(j)$, we cannot assume that $r_j \leq r_k$.

\begin{lemma}\label{lm:rad-rk}
  Suppose client $k \in \mC$ is relocated to $j \in C$ after filtering ($k \in D(j)$).
  Then $\rho_j \leq 3r_k$.
\end{lemma}

\begin{proof}
  Note that $y(B(k,r_k)) \ge 1$ via the LP constraint. We have $d(j,k) \le 2 \la(k) \le 2r_k$
  since $\la(k) = \min\{r_k,2\bC_k\}$. Via triangle inequality, $B(k,r_k) \subseteq B(j,3r_k)$.
  Thus $\rho_j \le 3r_k$.
\end{proof}

\Cref{lm:cost-gambound} and \Cref{lm:rad-rk} imply
\Cref{lm:rad-half}, which bounds the distance between relocated points
and the primary and secondary facilities of the cluster center they
are relocated to.


\begin{lemma}\label{lm:rad-half}
Let $k \in \mC$ and $k \in D(j)$ for a cluster center $j \in C$. Then, $d(j,
\io(j)) \leq \la(j) \leq \la(k) \leq r_k$. When $j \in C_s$, $d(j,
\it(j)) \leq \rho_j \leq 3r_k$. When $j \in C_b$, $d(j, \it(j)) \leq d(j, \sigma(j)) + d(\io(\sigma(j)),\sigma(j)) \leq 
3\gamma_j \leq 3\rho_j \leq 9r_k$.
\end{lemma}

\remark{Notice that the $v(B_j) = 1$ constraint imposed for points $j \in C_s$ ultimately did not effect the cost analysis in \Cref{lm:cost-4bound}. That is, we did not need to draw a distinction between points in $C_s$ and points in $C_b$ in order to obtain $COST'(\hx,\hy) \leq 4 \cdot COST'(x,y)$. The purpose of defining sets $C_s$ and $C_b$ and imposing an additional constraint for points in $C_s$ is to ensure certain radius guarantees. In particular, \Cref{lm:rad-half} would not hold if the constraint $v(B_j) = 1$ for $j \in C_s$ was not enforced in $\Q$.}

\subsection{Stage 3: Converting to an Integral Solution} \label{sec:pmatmed-int}
The procedure to convert the half-integral $(\hx,\hy)$ to an integral
solution involves setting up a matroid intersection instance consisting of
the input matroid $\M$ and a partition matroid that is constructed using the primary and
secondary facilities from $(\hx,\hy)$ after another clustering step.
The solution to this instance will be used
to construct an integral solution $(\tx,\ty)$ to $\cI'$.

For $j \in C$ set $\hC_j = (d(\io(j),j) +
d(j,\sigma(j)) + d(\it(j),\sigma(j)))/2$. In cases where $j$ has no
secondary facility, let $\it(j) = \io(j)$. For each $j \in C$, define
$S_j = \{i \mid \hx_{ij} > 0 \} = \{ \io(j), \it(j) \}$. $S_j$ has
either one or two facilities. In addition, the following holds
and will be relevant later.

\begin{claim}\label{clm:sj-cases}
When $S_j \cap S_{j'} \neq \emptyset$, one of three cases can
occur. (i) $S_j \cap S_{j'} = \{ \io(j), \it(j) \}$, in which case
$\sigma(j) = j'$ and $\sigma(j') = j$; (ii) $S_j \cap S_{j'} = \{
\io(j)\}$, and thus $\sigma(j') = j$ and $\sigma(j) \neq j'$ (a symmetric
case occurs when switching $j$ and $j'$); (iii) $S_j \cap S_{j'} =
\{\it(j)\}$ where $\it(j) = \it(j')$, hence $\sigma(j) = \sigma(j') =
p$ and $p \neq j,j'$.
\end{claim}

We construct a partition matroid $\N$ on ground set $\F$ via another
clustering process to create a set $C' \subseteq C$.  Repeat the
following two steps until no clients in $C$ are left to consider: (1)
Pick $j \in C$ with the smallest $\hC_j$ value and add $j$ to the set
$C'$ then (2) remove every $j' \in C$ where $S_j \cap S_{j'} \neq
\emptyset$, and have $j$ be the center of $j'$ (denoted by $\ctr(j') =
j$). It is easy to see that the sets $S_j, j \in C'$ are mutually disjoint.
Thus, a partition of $\F$ is induced by $\{S_j \mid j \in C'\}$, and
the set $\F \setminus \cup_{j \in C'} S_j$.
Set the capacity for each set of this
partition to $1$.

Now we consider the polytope that is intersection of the matroid polytopes of
$\M$ and $\N$:
\[ \R = \{z \in \Rset^\F_+ \mid \forall S \subseteq \F:~z(S) \leq r(S),~~\forall j \in C':~ z(S_j) \leq 1 \}\]
The polytope $\R$ has integral extreme points via the classical result of
Edmonds \cite{edmonds2003submodular,schrijver2003combinatorial}.

The goal now is to figure out the set of facilities to open by
optimizing a relevant objective over $\R$. First, we define a vector $\hy' \in \Rset^\F_+$:
if $i \in S_j$ for some $j \in C'$ we set $\hy'_i = \hx_{ij} \leq \hy_i$, otherwise
we set $\hy'_i = \hy_i$. Observe that $\hy'$ is feasible for $\R$ and shows that $\R$ is not empty.

We now define a linear function $H(\cdot)$ over vectors in $\R$.
We will optimize $H(\cdot)$ over $\R$ to obtain an integral extreme
point $\ty$ and we will analyze its cost via $\hy'$.
For $z \in \Rset_+^\F$, define $H(z)$ as follows.

\[ H(z) = \sum_i f_iz_i + \sum_{j \in C} L_{j}(z)\textnormal{, where}\]

\[ L_{j}(z) = \begin{cases}
      \sum_{i \in S_{\ctr(j)}}  a'_{j} d(i,j)z_i & \io(j) \in S_{\ctr(j)} \\
      \sum_{i \in S_{\ctr(j)}}  a'_{j} \Big(d(j, \sigma(j)) + d(\sigma(j),i)\Big)z_i  \\ \tab + a'_{j}  \Big(d(\io(j),j) - d(j,\sigma(j)) - d(\io(\sigma(j)),\sigma(j)) \Big)z_{\io(j)} & \textnormal{otherwise}
   \end{cases} \]

Let $\ty \in \R$ be an integer extreme point such that $H(\ty) \leq
H(\hy')$. We use this to define an integral solution $(\tx,\ty)$
to the modified instance by assigning each $j \in C'$ to the facility
opened from $S_j$ i.e. the facility $i \in S_j$ such that $\ty_i =
1$. For each $j' \in C \setminus C'$, assign $j'$ to either $\io(j')$
if it is open or the facility opened from $S_{\ctr(j')}$. $L_{j}(\ty)$ serves as a proxy and upper bound for $j$'s assignment cost. When $\io(j) \notin S_{\ctr(j)}$, the second term of $L_j(\ty)$ will adjust the distance $j$ pays depending on whether $\io(j)$ is opened or not. This adjustment is not needed when $\io(j) \in S_j$ or when $\io(j) \notin S_j$ is not opened, since in this case $j$ must be assigned to the center opened from $S_{\ctr(j)}$. The
following lemmas will show how the cost of $(\tx,\ty)$ can
be bounded by that of the half-integral solution $(\hx,\hy)$
from the previous stage.

\begin{lemma} \label{lm:cost-int-1}
  $COST'(\tx,\ty)$ is at most $H(\ty) \leq H(\hy')$.
\end{lemma}

\begin{proof}
   Since the facility costs of $(\tx,\ty)$ will remain as they are in $H(\ty)$, it suffices to show that for all $j \in C$, the assignment cost of $j$ is at most $L_j(\ty)$. When
   $j \in C'$, $\ctr(j) = j$ and the assignment cost of $j$ will be exactly $L_j(\ty)$. 
   
   Now we consider two possibilities for $j' \in C \setminus C'$. Let $\ctr(j') = j$. If $j'$ gets assigned to a center from $S_j$, there are two possible cases for the value of $L_{j'}(\ty)$. If $\io(j') \in S_j$ then the assignment cost for $j'$ is exactly $L_{j'}(\ty)$. Otherwise, $\io(j') \notin S_j$ and $\ty_{\io(j')} = 0$. In this case $L_{j'}(\ty) = \sum_{i \in S_j} a'_{j'}(d(j',\sigma(j)) + d(i,\sigma(j)))\ty_i$. By triangle inequality, $d(i,j') \leq d(i,\sigma(j')) + d(j,\sigma(j'))$, therefore the assignment cost of $j'$ is at most $L_{j'}(\ty)$.

   If $j'$ is assigned to a center that is not from $S_j$, it is because $\ty_{\io(j')} = 1$ and $\io(j') \notin
   S_j$. Here, the assignment cost of $j'$ is $a'_{j'}d(\io(j'),j')$. Let $i \in S_j$ be such that $\ty_i = 1$. The value of $L_{j'}(\ty)$ is therefore
   \begin{align*}
       L_{j'}(\ty) &= a'_j \Big(d(j',\sigma(j')) + d(i,\sigma(j')) + d(\io(j'),j) - d(j',\sigma(j')) - d(\io(\sigma(j)),\sigma(j))\Big) \\
       &= a'_j\Big(d(i,\sigma(j'))+ d(i,\io(j')) - d(\io(\sigma(j)),\sigma(j))\Big)
   \end{align*}
   Since $i \in S_j$ cannot be closer to $\sigma(j')$ than the primary facility of $\sigma(j')$, we know that $d(\io(\sigma(j)),\sigma(j)) \leq d(i,\sigma(j'))$. Thus, the assignment cost of $j'$ is at most $L_{j'}(\ty)$.
   \end{proof}


  \begin{lemma}\label{lm:cost-int-2}
   $H(\hy') \leq T(\hy)$.
   \end{lemma}


   \begin{proof}
  For notational ease, let $Q_j(\hy) := 2 \sum_{i \in G_j} d(i,j)\hy_i +  4\gamma_j(1-\sum_{i \in G_j} \hy_i)$. Thus, $T(\hy) = \sum_i f_i\hy_i + \sum_{j \in C} a'_j Q_j(\hy)$. As in the proof of the previous lemma, we focus on just the assignment costs of clients, since clearly $\sum_i f_i\hy_i' \leq \sum_i f_i \hy_i$. Specifically, we will show that $L_j(\hy') \leq a'_j Q_j(\hy)$ for all $j \in C$. For the remainder of the proof, we omit the term $a'_j$ from both sides of this inequality, since it remains fixed throughout our analysis. 

  First, we show $\hC_j \leq Q_j(\hy)$ for all $j \in C$. Recall that $j$ has no secondary facility when $\hy_{\io(j)} = 1$, in which case $\it(j) = \io(j)$. When $\hy(G_j) = 1$, $\sigma(j) = j$ and the primary and secondary facilities of $j$ are the only facilities in $G_j$ where $\hy_i > 0$. Since $\hy$ is half integral, we get $\hC_j = (d(\io(j),j) + d(\it(j),j))/2  = \sum_{i \in G_j} d(i,j)\hy_i \leq  Q_j(\hy)$. When $\hy(G_j) = 1/2$, $\sigma(j) = \ell \neq j$ and $\it(j) = \io(\ell)$. In this case $\hC_j = (d(\io(j),j) + d(j,\ell) + d(\io(\ell),\ell))/2$. Using \Cref{clm:sig-2} and definitions, $d(j,\ell) + d(\ell, \io(\ell)) \leq 3\gamma_j$. Therefore $\hC_j \leq \sum_{i \in G_j} d(i,j)\hy_i + 3\gamma_j(1-\hy(G_j)) \leq Q_j(\hy)$. 
  To prove $L_j(\hy') \leq a'_j Q_j(\hy)$ we consider several cases.
  

   \begin{enumerate}
       \item $j \in C'$: we have $\ctr(j) = j$ and $\io(j) \in S_j$. \begin{align*}
           L_j(\hy') &= \sum_{i \in S_j} d(i,j)\hy'_i \leq \sum_{i \in S_j} d(i,j)\hy_i  = \frac{1}{2}(\Big(d(\io(j),j) + d(\it(j),j)\Big) \\&\leq \frac{1}{2}\Big(d(\io(j),j) + d(\io(j),\sigma(j)) + d(\it(j),\sigma(j))\Big) \quad \text{(via triangle ineq.)} \\ &= \hC_j \leq Q_j(\hy).
       \end{align*}
       \item $j' \in C \setminus C'$. Let $\ctr(j') = j$. We have $\hC_j \leq \hC_{j'}$. 
       \begin{enumerate}
           \item $\io(j') \in S_j$. Then $\it(j) = \io(j')$ hence $\sigma(j) = j'$. \begin{align*} L_{j'}(\hy') &= \frac{1}{2} (d(\io(j),j') + d(\it(j),j')) \\ &\leq \frac{1}{2}  (d(\io(j),j) + d(j, j') + d(\it(j), j'))\quad \text{(via triangle ineq.)} \\ 
           &= \hC_{j} \le \hC_{j'}  \leq Q_{j'}(\hy). \end{align*}
           \item $\io(j') \notin S_j$: Then $S_j \cap S_{j'}$ is either $\{\io(j)\}$ or $\{ \it(j)\}$ (\Cref{clm:sj-cases}). In both cases, $\sigma(j') = \ell \neq j'$ and therefore $\hy(G_{j'}) = \hy_{\io(j')} = 1/2$. 
           Hence \[ L_{j'}(\hy') = \frac{1}{2} \cdot \Big(2 d(j',\ell) + d(\io(j),\ell) + d(\it(j),\ell) + d(\io(j'),j') - d(j',\ell) - d(\io(\ell),\ell)\Big)\] \begin{enumerate}
               \item When $S_j \cap S_{j'} = \{\io(j)\}$, $\io(j) = \it(j')$ thus $\ell = j$. Using the fact that $d(\it(j),j) \leq 2\hC_j - d(\io(j),j)$,  we have 
               \begin{align*}
                   L_{j'}(\hy') &= \frac{1}{2} \Big(2d(j',j) + d(\io(j),j) + d(\it(j),j) + d(\io(j'),j') - d(j',j) - d(\io(j),j)\Big) \\ 
                   &= \frac{1}{2} \Big(d(j,j') + d(\it(j),j) +  d(\io(j'),j')\Big) \\
                   &\leq \frac{1}{2} \Big( d(j,j')  + 2\hC_j - d(\io(j),j) + d(\io(j'),j') \Big) \\
                   &\leq \frac{1}{2} \Big( d(j,j')  + 2\hC_{j'} - d(\io(j),j) + d(\io(j'),j') \Big) \\
                   &= d(j,j') + d(\io(j'),j').
                 \end{align*}  
               \item When $S_j \cap S_{j'} = \{ \it(j) \}$, $\it(j) = \it(j') = \io(\ell)$ and so $\ell \neq j,j'$ and $\sigma(j) = \sigma(j') = \ell$. Since $2\hC_j \leq 2\hC_{j'}$, $d(\io(j),j) + d(j, \ell) \leq d(\io(j'),j') + d(j', \ell)$. Therefore,
               \begin{align*}
               L_{j'}(\hy') &= \frac{1}{2} \Big(d(j',\ell) + d(\io(j),\ell) + d(\io(j'),j')\Big) \\
               &\leq \frac{1}{2} \Big(d(j',\ell) + d(\io(j),j) + d(j,\ell) + d(\io(j'),j') \Big) \quad \text{(via triangle ineq.)} \\ &\leq \frac{1}{2} \Big(d(j',\ell) + d(\io(j'),j') + d(j',\ell) + d(\io(j'),j') \Big) \leq d(\io(j'),j') + d(j', \ell).
               \end{align*}
           \end{enumerate}
           Thus, in both cases we have
           \begin{align*}
                 L_{j'}(\hy') &\leq d(\io(j'),j') + d(j', \ell) \\
              &\leq d(\io(j'),j') + 2\gamma_{j'}  \quad \text{(via \Cref{clm:sig-2})} \\
               &\leq  2 \sum_{i \in G_{j'}} d(i,j')\hy_i + 4\gamma_{j'}(1-\sum_{i \in G_{j'}} \hy_i) = Q_{j'}(\hy) \quad \text{(since $\hy(G_{j'}) = 1/2)$}.
           \end{align*}
       \end{enumerate}
       This finishes the case analysis and the proof.
   \end{enumerate}
   \end{proof}
   

    \remark{We do not lose a factor in the cost when converting the half-integral solution to an integral solution because the analysis in Stage 2 ``overpays'' for the half-integral solution. We follow the approach from \cite{swamy16matroid}.}
\subsection{Cost and Radius Analysis for $\pmatmed$}

\Cref{lm:cost-filter,lm:cost-4bound,lm:cost-int-1,lm:cost-int-2} together imply the following bound on the cost of $(\tx,\ty)$ for instance $\cI$ with respect to the cost of the LP solution $(x,y)$. 
\begin{theorem} \label{thm:cost}
  $COST(\tx,\ty) \leq 12 \cdot COST(x,y)$. 
\end{theorem}

\begin{proof}
  $COST'(\tx,\ty)$ will be at most $T(\hy)$
  (\Cref{lm:cost-int-1,lm:cost-int-2}), and $T(\hy)$ is at most $4
  \cdot COST'(x,y) \leq 8\cdot COST(x,y)$
  (\Cref{lm:cost-4bound}). Hence, $(\tx, \ty)$ will give a
  solution to $\cI'$ of cost at most $8\cdot
  COST(x,y)$. \Cref{lm:cost-filter} tells us that translating an integer solution for $\cI'$ to an integer solution for $\cI$ 
  will incur an additional cost of at most $4\cdot COST(x,y)$. All
  together, $COST(\tx,\ty) \leq COST'(\tx,\ty) + 4\cdot COST(x,y) \leq 8 \cdot COST(x,y) + 4\cdot COST(x,y) = 12\cdot COST(x,y)$.
\end{proof}

To complete our analysis of the radius approximation factor, we must
determine how far points will be made to travel once the final centers
are chosen. In \Cref{lm:rad-half} we guaranteed that each cluster
center $j$ will not travel farther than $3\rho_j$ to reach its secondary
facility. However, in this final stage, we are assigning some cluster
centers to others, and cannot guarantee that their primary or
secondary facility will be opened. We can still show that even if a
cluster center $j$ from $C_s$ gets assigned to a cluster center $\ell$
from $C_b$ (i.e. that $\ctr(j) = \ell$), $j$ will still only travel a
constant factor outside of $\rho_j$. Consequently, using \Cref{lm:rad-rk}
we can show that each client $k \in \mC$ will travel only a constant
factor times its radius value $r_k$.

\begin{lemma} \label{lm:rad-int}
  Let $k \in \mC$, where $k \in D(j)$ for $j \in C$.
  The final solution will open a facility $i$
  such that $d(i,j) \le 19r_k$.
\end{lemma}

\begin{proof}
There are several cases to consider but most of them are simple.  We
provide the analysis for the case that gives the $19$ factor, and
other notable cases.

If $j \in C'$, then either $\io(j)$ or $\it(j)$ will be opened in the
final solution. \Cref{lm:rad-half} indicates that $j$ will be assigned
to a center that is at most $9r_k$ away. If $j \notin C'$, it must be
the case that $\ctr(j) = \ell$ where $S_\ell \cap S_j \neq \emptyset$,
and $\hC_{\ell} \leq \hC_{j}$. We claim that $\hC_j \leq \frac12(d(i_1(j),j) + d(i_2(j),j))
\le \frac{1}{2}(r_k + 9r_k) = 5r_k$ where we used \Cref{lm:rad-half}
to bound $d(i_1(j),j)$ and $d(i_2(j),j))$.

The farthest that $j$ would have to travel occurs when $j$ and $\ell$
share secondary facilities, and $\ell$'s primary facility is
opened (see \Cref{fig:ljp}). More precisely, this is when $S_\ell \cap S_j = \{ \it(\ell)
\} = \{ \it(j) \}$ and $\sigma(\ell) = \sigma(j) = p$ where $p$ is not
$j$ or $\ell$, and $\io(\ell)$ is opened at the end of Stage 3. 
\begin{figure}
    \centering
    \includegraphics[width = \textwidth]{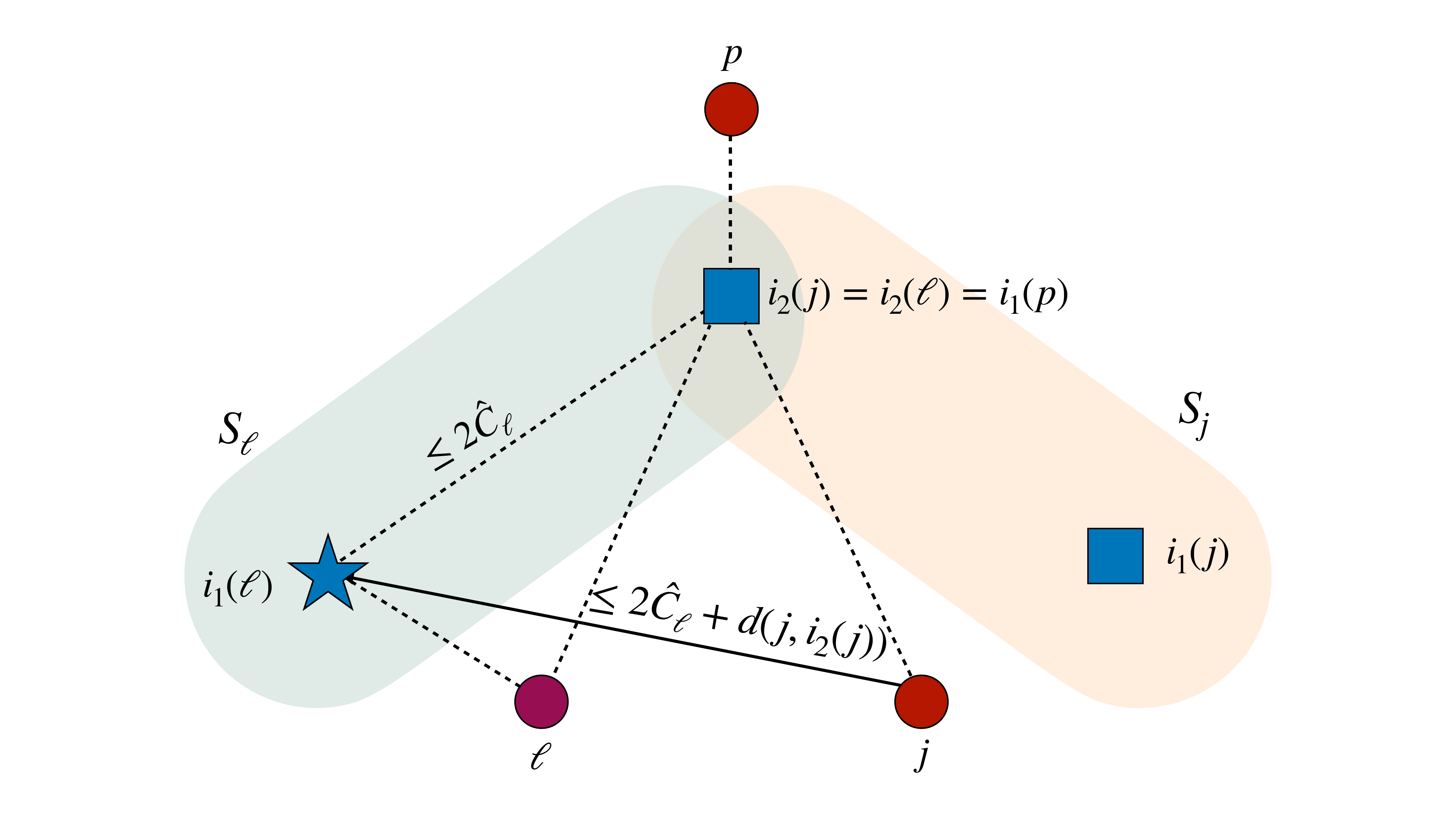}
    \caption{The farthest a point $j \in C$ will be from an opened center occurs when $\ctr(j) = \ell$, $\sigma(j) = \sigma(\ell) = p$, and $\io(\ell)$ is opened. }
    \label{fig:ljp}
\end{figure}
In this case, we have
\begin{align*}
  d(\io(\ell), j) &\leq d(\io(\ell),\it(\ell)) + d(\it(\ell),j) \leq d(\io(\ell),\ell) + d(\it(\ell),\ell) + d(\it(j),j) \\
  &= 2\hC_\ell + d(\it(j),j) \leq 2\hC_j + d(\it(j),j) \leq 10r_k + 9r_k = 19r_k.
\end{align*}
\end{proof}


\remark{Notice that in the last step of our proof for \Cref{lm:rad-int}, we bound the distance $d(\io(\ell), j)$ by $d(\io,j) + 2d(\it(j),j)$, where $d(\it(j),j) \leq 9r_k$. Hence, the majority of the distance that $j$ is traveling, according to our analysis, is due to the distance between $j$ and its secondary facility. If we could guarantee that cluster center $j$ has a reasonably close secondary facility, we could improve this radius factor. We will explore this further in \Cref{sec:upmatmed-2}.}\label{rem:sec}

Using \Cref{lm:rad-filter,lm:rad-int}, we have the following radius bound for the output of our algorithm. 

\begin{theorem} \label{thm:radius}
  Let $S$ be the output of the aforementioned approximation algorithm. For all $k \in \mC$, $d(k,S) \leq 21r_k$.
\end{theorem}

\Cref{thm:cost,thm:radius} together give us \Cref{thm:main-1}.
\section{Uniform Priority Matroid Median} \label{sec:upmatmed}

The $\upmatmed$ problem is a special case of the $\pmatmed$ problem in which all clients have the same radius value $L$. 
An instance $\cJ$ of the $\upmatmed$ problem can be described using the tuple $(\F,\mC, d,\f,L,\a,\M)$. We will abuse notation and interpret $L$ as not only a single radius value, but also as a vector from $\Rset^\mC$ where each entry is $L$; this will allow us to use our algorithm for $\pmatmed$ on instances of $\upmatmed$. 

In this section we show how we can take advantage of the uniform radius requirement to improve upon the $(21,12)$-approximation for $\pmatmed$.
In particular, since
we have $\bC_j \le L$ for all $j \in \mC$, we can pick points in filtering in order of their $\bC_j$ values and set $\phi(j) := \bC_j$ for $\filter$. This setting of $\phi$ will be compatible with the setting of $\la(j) := \min\{ L,2\bC_j\}$. Furthermore, $\filter$ with these $\phi$ and $\la$ functions is identical to the filtering step in Kamiyama's algorithm \cite{kamiyama20distance}. Notice that these same settings for $\pmatmed$, i.e. $\phi(j) := \bC_j$ and $\la := \min\{r_j, 2\bC_j \}$, are not necessarily compatible. The uniform radius constraint also help us to derive tighter bounds throughout the radius analysis of the $\pmatmed$ algorithm. 

Using the above observations, our algorithm for $\upmatmed$ is the following. For an instance $\cJ$ of $\upmatmed$, run the algorithm from \Cref{sec:pmatmed}, this time using $\phi(j) := \bC_j$ and $\la(j) := \min\{ L, 2\bC_j \}$ in the filtering stage (\Cref{sec:pmatmed-filter}) to construct instance $\cJ'$. Thus, the only change in the algorithm is the filtering step. We argue that this algorithm yields a better approximation algorithm for $\upmatmed$.

\begin{theorem}[\Cref{thm:upmatmed}a]\label{thm:upatmed}
  There is a $(9,8)$-approximation algorithm for $\upmatmed$.
\end{theorem}

\subsection{Cost and Radius Analysis for $\upmatmed$}
Since our algorithm for $\upmatmed$ only slightly differs from the one in \Cref{sec:pmatmed}, we omit several proofs that apply here. The only change to the cost analysis occurs in the filtering stage (\Cref{sec:pmatmed-filter}). In particular, we can derive a tighter bound than in \Cref{lm:cost-filter}. This ultimately leads to the improved cost bound, shown in \Cref{thm:cost-uni}.

\begin{lemma} \label{lm:cost-filter-uni} The following is true of
  $\cJ'$. (a) $COST'(x,y) \leq COST(x,y)$. (b) Any integer
  solution $(x',y')$ for $\cJ'$ can be converted to an integer
  solution for $\cJ$ that incurs an additional cost of at most
  $4\cdot COST(x,y)$.
\end{lemma}

\begin{proof}
Consider a non-center client $k \in \mC \setminus C$, where $k \in D(j)$ for $j \in C$ after \Cref{alg:filter}. 
Observe that in both instances, the cost of opening facilities ($\sum_{i} f_i y_i$) stays the same. In the original instance, the cost of assigning clients is $\sum_j \bC_j$. In the modified instance, $k$ is now paying $\bC_j$ instead of $\bC_k$. For $k$ to be assigned to $j$ in the filtering step, it must be the case that $\phi(j) = \bC_j \leq \bC_k = \phi(k)$. Therefore $COST'(x,y) \leq  COST(x,y)$. The proof for the second part is the same as in \Cref{lm:cost-filter} since $\la(k) = \min\{L,2\bC_k\}$ and $d(j,k) \le 2\la(k)$.
\end{proof}

Recall that $(\tx,\ty)$ is the final integeral solution output by the algorithm. 
\begin{theorem} \label{thm:cost-uni}
  $COST(\tx,\ty) \leq 8 \cdot COST(x,y)$. 
\end{theorem}

\begin{proof}
  $COST'(\tx,\ty)$ will be at most $T(\hy)$
  (\Cref{lm:cost-int-1,lm:cost-int-2}), and $T(\hy)$ is at most $4
  \cdot COST'(x,y) \leq 4 \cdot COST(x,y)$
  (\Cref{lm:cost-4bound,lm:cost-filter-uni}). \Cref{lm:cost-filter-uni} also tells us that translating an integer solution for $\cJ'$ to an integer solution for $\cJ$ 
  will incur an additional cost of at most $4\cdot COST(x,y)$. Together, $COST(\tx,\ty) \leq COST'(\tx,\ty) + 4 \cdot COST(x,y) \leq 4 \cdot COST(x,y) + 4\cdot COST(x,y) = 8\cdot COST(x,y)$.
\end{proof}

We now analyze the radius guarantee and outline the changes in the analysis.
First, we have the following lemma in place of \Cref{lm:rad-filter} which also follows directly from \Cref{fact:filter}. 
\begin{lemma}\label{lm:rad-filter-uni}
Let $k \in \mC$ be assigned to $j \in C$ after using the Filtering procedure (i.e. $k \in D(j)$). Then, $d(j,k) \leq 2 \la(k) \leq 2L$. 
\end{lemma}

Since all radius values are equal, we do not need \Cref{lm:rad-rk} to relate the radius values of different clients. We do need to update \Cref{lm:rad-half} and \Cref{lm:rad-int}. These updated lemmas are given below. 

 \begin{lemma}
 \label{lm:rad-half-uni}
   Let $k \in \mC$ and $k \in D(j)$ for a cluster center $j \in C$. Then (a) $d(j, \io(j)) \leq \la(j) \leq \la(k) \leq L$ and (b)  when $j \in C_s$ $d(j, \it(j)) \leq \rho_j \leq L$ and (c) when $j \in C_b$ $d(j, \it(j)) \leq d(j, \sigma(j)) + d(\io(\sigma(j)),\sigma(j)) \leq 3\gamma_j \leq 3\rho_j \leq 3L$.
   \end{lemma}

The reasoning for the preceding lemma is the same as \Cref{lm:rad-half}, except $L$ is used in place of $r_j$ and $r_k$ values. 


\begin{lemma} \label{lm:rad-int-uni}
Let $j \in C$. The final solution will open a facility $i$ such that $d(i,j) \leq 7L$. 
\end{lemma}

\begin{proof}
Notice that the farthest $j$ will have to travel is when $j \notin C'$ and $\ctr(j) = \ell \neq j$. 

 If $j \in C'$, then $j$ will be assigned to a center that is at most $3L$ away (\Cref{lm:rad-half-uni}). If $j \notin C'$, then $\ctr(j) = \ell$ hence $S_\ell \cap S_j \neq \emptyset$, and $\hC_{\ell} \leq \hC_{j}$. Now we will have $\hC_j \leq \frac{1}{2}(L + 2L + L) = 2L$. 

In the worst case, $S_\ell \cap S_j = \{ \it(\ell) \} = \{ \it(j) \}$ and $\sigma(\ell) = \sigma(j) = p$ where $p$ is not $j$ or $\ell$, and $\io(\ell)$ is ultimately opened in the final solution. Using reasoning similar to the proof of \Cref{lm:rad-int}, we have 
\begin{align*}
    d(\io(\ell), j) &\leq  2\hC_\ell + d(\it(j),j) \leq 2\hC_j + d(\it(j),j) \leq 4L + 3L = 7L.
\end{align*}
   
\end{proof}

\Cref{lm:rad-filter-uni} and \Cref{lm:rad-int-uni} give us the following improved radius bound for the solution output by the algorithm. This, along with \Cref{thm:cost-uni}, proves \Cref{thm:upatmed}.

\begin{theorem} \label{thm:radius-uni}
  Let $S$ be the output of the aforementioned approximation algorithm for $\upmatmed$. For all $k \in \mC$, $d(k,S) \leq 9L$. 
\end{theorem}

\remark{Our results imply that we can obtain a $(9,8)$-approximate solution for ``well-behaved'' instances of $\pmatmed$ (where radii are not necessarily uniform) in which the new settings of $\phi$ and $\la$ are still compatible. We formalize this observation. Suppose an instance $\I$ of $\pmatmed$ satisfies the properties that (i) the functions $\phi(j) := \bC_j$ and $\la(j) := \min\{ 2\bC_j, r_j\}$ are compatible and (ii) for all $j,k \in \mC$ such that $\phi(j) \leq \phi(k)$, $r_j \leq r_k$. Then, a $(9,8)$-approximation is achievable for $\I$.
}

\section{Analysis for $(36,8)$-approximation for $\pmatmed$}\label{sec:pmatmed-2}

In this section we show how to obtain a $(36,8)$-approximate solution for $\pmatmed$. Our algorithm is as follows. Run the algorithm from \Cref{sec:pmatmed} using $\phi(j) := \bC_j$ and $\la(j) := 2\bC_j$ in $\filter$. Clearly, $\phi$ and $\la$ are compatible. Furthermore, notice that this setting of $\phi$ is identical to that of our algorithm of $\upmatmed$. Since cost analysis for the filtering stage of $\upmatmed$ only uses $\phi$ (and not $\la$), \Cref{lm:cost-filter-uni} and \Cref{thm:cost-uni} hold in this case as well. This is the reason why the cost factor guarantee will be $8$.

Though our setting for $\la$ does not use radius values, from the $\pmatmed$ LP constraint, $\forall j \in \mC$, $\bC_j \leq r_j$ holds. Therefore, $\la(j) = 2\bC_j \leq 2r_j$. Previous settings of $\la$ (where $\la(j) := \min\{r_j, 2\bC_j\}$) were such that $\la(j) \leq r_j$. Thus the new setting of $\la$ can lead to a weakening of the radius guarantee. First, we formalize the above observation in \Cref{fact:filter-tr} which we will use to update the radius analysis of \Cref{sec:pmatmed}. 

\begin{fact} \label{fact:filter-tr}
The following holds after $\filter$ when $\phi(j) := \bC_j$ and $\la(j) := 2\bC_j$: (a) $\bC_j \leq r_j$, and hence $\la(j) = 2\bC_j \leq 2r_j$, (b) $\forall k \in D(j) ~~d(j,k), \leq 2\la(k) \leq 4C_{k} \leq 4r_{k}$.
\end{fact}

The following updated lemmas now hold in place of their counterparts from \Cref{sec:pmatmed}. The proofs for these results are identical to those from \Cref{sec:pmatmed} up to certain bounds that change due to the above fact and the subsequent lemmas. These changes occur whenever definitions of $\phi$ and $\la$ are used in the analysis, and the following lemmas will be invoked in place of their counterparts from \Cref{sec:pmatmed}.

\begin{lemma} [Updated \Cref{lm:rad-filter}] \label{lm:rad-filter-tr}
Let $k \in \mC$ be assigned to $j \in C$ after using the Filtering procedure (i.e. $k \in D(j)$). Then, $d(j,k) \leq 2\la(k) \leq 4r_k$.
\end{lemma}

\begin{lemma} [Updated \Cref{lm:rad-rk}] \label{lm:rad-rk-tr}
For some $k \in \mC$, where $k \in D(j)$, $\rho_j \leq 5r_k$.
\end{lemma}

\begin{lemma}[Updated \Cref{lm:rad-half}]\label{lm:rad-half-tr}
Let $k \in \mC$ where $k \in D(j)$ for $j \in C$. (a) $d(j, \io(j)) \leq \la(j) \leq \la(k) \leq 2r_k$, (b) when $j \in C_s$, $d(j, \it(j)) \leq \rho_j \leq 5r_k$, and (c) when $j \in C_b$, $d(j, \it(j)) \leq d(j, \sigma(j)) + d(\io(\sigma(j)),\sigma(j)) \leq 3\gamma_j \leq 3\rho_j \leq 15r_k$.
\end{lemma}

\begin{lemma}[Updated \Cref{lm:rad-int}] \label{lm:rad-int-tr}
Let $k \in \mC$, where $k \in D(j)$ for $j \in C$. The final solution will open a facility $i$ such that $d(i,j) \leq 32r_k$. 
\end{lemma}

Finally, using \Cref{lm:rad-filter-tr} \Cref{lm:rad-int-tr}, along with \Cref{thm:cost-uni}, we get the following result. 

\begin{theorem}[\Cref{thm:main-1}(b)]
  There is a $(36,8)$-approximation algorithm for Priority Matroid Median.
\end{theorem}

\section{Tighter Radius Guarantee for \upmatmed} \label{sec:upmatmed-2}

In this section we will show how to obtain the following tighter radius guarantee for instances of $\upmatmed$. 

\begin{theorem}[\Cref{thm:upmatmed}(b)] \label{thm:upmatmed-imp}
  For any fixed $\epsilon > 0$ there is a $(5+8\epsilon,4 + \frac{2}{\epsilon})$-approximation for $\upmatmed$. 
\end{theorem}

In the previous result for $\upmatmed$, we set $\phi(j) := \bC_j$ and $\lambda := \min\{L, 2\bC_j\}$. In the second result for $\pmatmed$, we showed how setting $\lambda(j) := 2L$ would increase the radius guarantee. Thus, in order to tighten the radius guarantee for $\upmatmed$, we will again change $\la(j)$, but this time in a way that will allow points to have tighter radius bounds.  

To build up to our new setting for $\lambda$, we first partition points in the original client set into points that have relatively small, or \emph{tiny} $\bC_j$ values, $\mC_T = \{ j \in \mC \mid \bC_j \leq \epsilon L\}$ and points that have \emph{large} $\bC_j$ values, $\mC_L = \{ j \in \mC \mid \bC_j > \epsilon L\}$. Now, our algorithm is as follows: Run \Cref{alg:app} on $\pmatmed$ instance $\cI$, but in Line 0, set $\phi(j) := \bC_j$ and $\la(j)$ as defined below. 

\[ \lambda(j) = \begin{cases}
      2\bC_j & j \in \mC_T \\
      L & j \in \mC_L
   \end{cases}\]

Note that $\phi$ and $\lambda$ will satisfy compatibility. Let $C_L$ and $C_T$ denote the subsets of cluster centers $C$ that belong to $\mC_L$ and $\mC_T$, respectively. Furthermore, we have the following result.

\begin{fact}
The following holds after running \Cref{alg:filter} when $\phi(j) := \bC_j$ and $\lambda(j) := \Lambda(j)$: Consider $j \in C$, and $k \in D(j)$. (a) If $j \in C_T$, then $k \in \mC_T$ and $k \in \mC_L$. (b) If $j \in C_L$, then $k \in \mC_L$. 
\end{fact}

In addition to altering the setting for $\lambda$, we will make another change to to the filtering stage (one that was not done in the previous sections) which allows us to improve radius guarantees slightly further than if we were to only change settings of $\lambda$ and $\phi$. We alter \Cref{ln:chld} of $\filter$ to set $D(j) := \{ k \in U \mid d(j,k) \leq \lambda(j) + \lambda(k)\}$. In making this change to $\filter$, note that \Cref{fact:filter}(b) no longer holds. Instead we have the following.

\begin{fact}\label{fact:filter-2}
The following is true after running the altered $\filter$ procedure. $\forall j, j' \in C$, $d(j,j') > \lambda(j) + \lambda(j')$. 
\end{fact}

As such, we must make sure that any result that utilizes \Cref{fact:filter}(b) it still holds. The parts of our analysis that use \Cref{fact:filter}(b) are located in Stages 2 and 3. In particular, we use it to prove \Cref{lm:cost-gambound}, which later gets used to prove \Cref{lm:cost-4bound}, \Cref{lm:rad-half-uni}, and \Cref{lm:cost-int-2}. Luckily, we can show that despite the fact that the bound from \Cref{lm:cost-gambound} becomes larger, the slack that existed in our analysis of the half-integral solution allows the cost to remain unchanged. We provide the details for the specific changes that must be made throughout our analysis below.

\subsection{Cost and Radius Analysis} 

We now provide the updated results that occur due to the change to $\lambda$ and the updated \Cref{ln:chld}, and show how this change leads to a $(5+8\epsilon,4 + 2/\epsilon)$-approximate solution. 

The first change that occurs for the cost analysis lies in the additional cost incurred when converting the integral solution of the updated instance $\cJ'$ to a solution for $\cJ$. In particular, we have the following updated version of \Cref{lm:cost-filter-uni}. Note that this change is the only change that is due to updating $\lambda$. 

\begin{lemma}\label{lm:cost-filter-uni-imp}
The following is true of $\cJ'$. (a) $COST'(x,y) \leq COST(x,y)$. (b) Any integer solution $(x',y')$ for $\cJ'$ can be converted to an integer solution for $\cJ$ that incurs an additional cost of at most $\frac{2}{\epsilon}\cdot COST(x,y)$. 
\end{lemma}

\begin{proof}
The proof of $(a)$ is identical to that from the proof of \Cref{lm:cost-filter-uni}. For $(b)$, Let $\beta$ denote the cost of an integer solution to $\cJ'$. Consider a client $k \in \mC \setminus C$ that is relocated to $j \in C$. If $j$ connects to a facility $i$ in the integer solution for $\cJ'$, $k$ can connect to $i$ in the solution to $\cJ$, and its per unit connection cost will increase by at most $d(j,k) \leq 2L$. Observe that in the worst case, $k \in \mC_L$, and hence $\bC_k > \epsilon L$. Therefore, $d(j,k) \leq 2L < \frac{2}{\epsilon} \bC_k$. Thus, the total increase in connection costs when compared to $\beta$ is upper bounded by $\sum_{j \in C} \sum_{k \in D(j)} a_k \cdot \frac{2}{\epsilon}\bC_k \leq \frac{2}{\epsilon}\cdot COST(x,y)$. 
\end{proof}

Now, to reconcile potential changes to the cost analysis that occur by altering \Cref{ln:chld}, we begin by stating an updated version of \Cref{lm:cost-gambound}. 

\begin{lemma}\label{lm:cost-gambound-2}
Consider some $j \in C$, and let $\gamma_j = d(i,j)$ where $i \in
F_{j'}$ for some ${j'} \neq j$. For every $i' \in F'_{j'}$, $d(i', j)
\leq 4\gamma_j$.
\end{lemma}

As was done in our original analysis, keeping the preceding lemma in mind we can use as proxy for $j$'s
per-unit-demand cost a function written in terms of the facility vector $v$. When
$y'(G_j) = 1$, the cost for $j$ can be bounded by
$\sum_{i \in G_j} d(i,j) y'_i \leq \bC_j$. When
$y'(G_j) < 1$, the preceding lemma indicates that we can upper bound the
cost of the solution by $\sum_{i \in G_j} d(i,j) y'_i +
4\gamma_j(1-y'(G_j)) \leq 4 \cdot \bC_j$. Recall that the function $T(\cdot)$ for $v \in \Q$ was defined as follows:


\[ T(v) = \sum_{i \in \F} f_i v_i + \sum_{j \in C} a'_j \Big( 2 \sum_{i \in G_j} d(i,j) v_i + 4\gamma_j(1-v(G_j))\Big) \]

For a vector $v$ such that $v(F'_j) \geq 0.5$ and $v(G_j) \leq 1$ for all $j \in C$, the term $a'_j (2\sum_{i \in G_j} d(i,j) y'_i + 4\gamma_j(1-y'(G_j)))$ will still upper bound $j$'s assignment cost with respect to $v$, now via \Cref{lm:cost-gambound-2}. When $v(G_j) = v(B_j) = 1$ for $j \in C_s$, $j$'s assignment cost will be at most $a'_j (2\sum_{i \in B_j} d(i,j) v_i)$.
Note that $T(v)$ remains an overestimate.

Now, we show that \Cref{lm:cost-4bound} will effectively remain the same (save for the last inequality, which will change because of \Cref{lm:cost-filter-uni-imp}(a)) by providing an updated proof. To avoid confusion, we will also rewrite the lemma below. This will also contain the proof of of \Cref{lm:cost-gambound-2}.

\begin{lemma}[Updated \Cref{lm:cost-4bound}] \label{lm:cost-4bound-2}
$COST'(\hx,\hy) \leq T(\hy) \leq T(y') \leq 4 \cdot COST'(x,y) \leq 4 \cdot COST(x,y)$.
\end{lemma}

\begin{proof}
  We first show that $T(y') \leq 4 \cdot COST'(x,y)$ (we already have $T(\hy) \leq T(y')$).
  We know that $COST'(x,y)$ can be expressed
as $\sum_i f_iy_i + \sum_j a'_j\cdot \bC_j$. For any $j \in C$,
observe that $\bC_j = \sum_{i \in G_j} d(i,j) x_{ij} + \sum_{i \notin
  G_j} d(i,j)x_{ij}$ and hence $\bC_j \geq \sum_{i \in G_j} d(i,j)
x_{ij} + \gamma_j \sum_{i \notin G_j} x_{ij}$.

\begin{align*}
     T(y') &\leq \sum_i f_iy_i + \sum_j a'_j \Big(2\sum_{i \in G_j} d(i,j)x_{ij} + 4\gamma_j\Big(1-\sum_{i \in G_j} x_{ij}\Big)\Big) \\
     &\leq \sum_i f_iy_i + 4 \sum_j a'_j \cdot \bC_j \leq 4 \cdot COST'(x,y)
\end{align*}

Next, we upper bound $COST'(\hx,\hy)$ by $T(\hy)$. It suffices to focus
on the assignment cost. Consider $j \in
C_s$. Its primary and secondary facilities are in $B_j$ and it is easy
to see that its connection cost is precisely $\sum_{i \in B_j}
d(i,j)\hx_{ij}$. Now consider $j \in C_b$. Recall that when $\hy(G_j) = 1$,
the total assignment cost of $j$ is at most $\sum_{i \in G_j} d(i,j)
\hy_i$.  When $\hy(G_j) < 1$, $j$ connects to primary facility in
$F'_j$ and a secondary facility. The second nearest facility will not be
in its $G_j$ ball, i.e. $\it(j) \notin F_j$.  Let $j' \neq j$ be
client that defines $\gamma_j$. Via \Cref{lm:cost-gambound-2}, we have
$d(\it(j), j) \leq 4 \gamma_j$. Assuming this, when $\hy(G_j) < 1$,
the total assignment cost of $j$ is at most $\sum_{i \in G_j} d(i,j)
\hy_i + 4\gamma_j (1-\hy(G_j))$. Based on these assignment cost
upper bounds we see that $COST'(\hx,\hy) \le T(\hy)$.

Now we prove \Cref{lm:cost-gambound-2}.
From \Cref{fact:filter-2} we have $\la(j) + \la(j') \leq d(j,j')$.
Via triangle inequality $d(j,j') \le d(j,i) + d(i,j') \leq 2\gamma_j$.
Thus $\la(j') \leq \la(j) + \la(j') \leq 2\gamma_j$ which implies that $\la(j') \leq 2\gamma_j$.
Recall that $F'_{j'}$, from its definition, is contained in a ball of radius $\la(j')$ around $j'$. Thus, for any facility $i' \in F'_{j'}$, $d(i',j') \leq \la(j')\le 2\gamma_j$,
Therefore, using \Cref{clm:sig-2}, $d(i',j) \leq d(j,j') + d(j',i') \leq 4\gamma_j$. This gives
us the lemma.

Finally, using \Cref{lm:cost-filter-uni-imp}, we know that $COST'(x,y) \leq COST(x,y)$, hence $4\cdot COST'(x,y) \leq 4 \cdot COST(x,y)$.
\end{proof}

Now, we claim that \Cref{lm:cost-int-2} remains the same, and provide an updated proof.

\begin{proof}[Update to Proof of \Cref{lm:cost-int-2}]
  For notational convenience, let $Q_j(\hy) := 2 \sum_{i \in G_j} d(i,j)\hy_i +  4\gamma_j(1-\sum_{i \in G_j} \hy_i)$. Thus, $T(\hy) = \sum_i f_i\hy_i + \sum_{j \in C} a'_j Q_j(\hy)$. We focus on just the assignment costs of clients, since clearly $\sum_i f_i\hy_i' \leq \sum_i f_i \hy_i$. Specifically, we will show that $L_j(\hy') \leq a'_j Q_j(\hy)$ for all $j \in C$. For the remainder of the proof, we omit the term $a'_j$ from both sides of this inequality, since it remains fixed throughout our analysis. 

  We will show $\hC_j \leq Q_j(\hy)$ for all $j \in C$. Recall that $j$ has no secondary facility when $\hy_{\io(j)} = 1$, in which case $\it(j) = \io(j)$. When $\hy(G_j) = 1$, $\sigma(j) = j$ and the primary and secondary facilities of $j$ are the only facilities in $G_j$ where $\hy_i > 0$. Since $\hy$ is half integral, we get $\hC_j = (d(\io(j),j) + d(\sigma(j),j) + d(\sigma(j),\io(\sigma(j))))/2  = \sum_{i \in G_j} d(i,j)\hy_i \leq  Q_j(\hy)$. When $\hy(G_j) = 1/2$, $\sigma(j) = \ell \neq j$ and $\it(j) = \io(\ell)$. In this case $\hC_j = (d(\io(j),j) + d(j,\ell) + d(\io(\ell),\ell))/2$. Using \Cref{lm:cost-gambound-2}, $d(j,\ell) + d(\ell, \io(\ell)) \leq 4\gamma_j$. Therefore $\hC_j \leq \sum_{i \in G_j} d(i,j)\hy_i + 4\gamma_j(1-\hy(G_j)) \leq Q_j(\hy)$.

  The remainder of the proof is showing that $L_j(\hy') \leq \hC_j$ for all possible cases and does not utilize \Cref{lm:cost-gambound-2}, it remains identical to that of the proof of \Cref{lm:cost-int-2}, which can be found in \Cref{sec:pmatmed-int}. 
\end{proof}

Recall that $(\tx,\ty)$ is the final integral solution output by the algorithm. The following theorem holds by following the reasoning from \Cref{thm:cost-uni}, but uses the result from \Cref{lm:cost-filter-uni-imp} in place of \Cref{lm:cost-filter-uni}. 

\begin{theorem}
$COST(\tx,\ty) \leq (4 + \frac{2}{\epsilon})\cdot COST(x,y)$
\end{theorem}

We will now analyze the radius factor. First, we have the following lemma in place of \Cref{lm:rad-filter-uni}. Observe that statement (b) of this lemma is a consequence of the alteration made to \Cref{ln:chld}, and would otherwise be $d(j,k) \leq 2L$ had this particular change not been made.

\begin{lemma}\label{lm:rad-filter-uni-imp2}
Let $k \in \mC$ be assigned to $j \in C$ after the filtering stage. Then we have three cases: (a) if $j,k \in \mC_T$, then $d(j,k) \leq \lambda(k) + \lambda(j) \leq 4\epsilon L$; (b) if $j \in \mC_T$ and  $k \in \mC_L$, then $d(j,k) \leq \lambda(j) + \lambda(k) \leq (1 + 2\epsilon)L$; (c) if $j,k \in \mC_L$, then $d(j,k) \leq \lambda(j) + \lambda(k) = 2L$.
\end{lemma}

Now, to conduct the radius analysis for the half-integral stage, we must make use of the following important result. 

\begin{claim}
(a) For all $j \in C_L$, $j \in C_s$. (b) For all $j \in C_T$, then $j \in C_s$ or $j\in C_b$. 
\end{claim}

\begin{proof}
Recall that we defined the sets $C_s = \{ j \in C \mid \rho_j \leq \gamma_j\}$ and $C_b = \{j \in C \mid \rho_j > \gamma_j$ \}. Also recall that per \Cref{fact:fjgj}, $\rho_j \leq L$, and $\gamma_j > \lambda(j)$. This means that for points $j \in C_L$, it must be the case that $\gamma_j > \lambda(j) = L \geq \rho$, and therefore $j$ must be in $C_s$. 

For $j \in C_T$, we only know that $\gamma_j > 2\bC_j$. Depending on how large this quantity is in relation to $\rho_j$, $j$ can be in either $C_s$ or $C_b$. 
\end{proof}

Recall that in our algorithm, points $j \in C_s$ will have both a primary and secondary facility opened within their $B_j$ balls. Since all points of $C_L$ will be in $C_s$, we are guaranteeing that all ``large'' cluster centers, i.e. cluster centers $j$ such that $\bC_j > \epsilon L$ will have primary and secondary facilities that are reasonably close. This will prevent these points from having to utilize facilities from other cluster centers, which is where a bulk of the radius loss occurs. 

Now, we have the following updated versions of \Cref{lm:rad-half-uni} and \Cref{lm:rad-int-uni}.

\begin{lemma}\label{lm:rad-half-uni-imp}
Let $k \in \mC$ and $k \in D(j)$ for cluster center $j \in C$. (a) $d(j, \io(j)) \leq \lambda(j)$ which is at most $2\epsilon L$ for $j \in \mC_T$, or $L$ for $j \in \mC_L$. (b) For $j \in C_s$, $d(j, \it(j)) \leq \rho(j) \leq L$. (c) For $j \in C_b$, $d(j, \it(j)) \leq d(j,\sigma(j)) + d(\sigma(j),\io(\sigma(j)) \leq 2L + 2\epsilon L$.
\end{lemma}

\begin{proof}
The reasoning for part (a) and (b) matches that of \Cref{lm:rad-half-uni}. For (c), it follows from \Cref{claim:cscb} that $d(j, \sigma(j)) \leq 2\gamma_j \leq 2L$. We will now argue that for $j \in C_b$, $\sigma(j) \in C_T$. This is because $d(j,\sigma(j)) \leq 2L$, and if it was the case that $\sigma(j) \in C_L$, then $\sigma(j)$ would have been added to $D(j)$ after filtering. Hence, $\sigma(j) \in C_T$, and therefore by (a), $d(\sigma(j),\io(\sigma(j))) \leq 2 \epsilon L$. 
\end{proof}

\begin{lemma}\label{lm:rad-int-uni-imp}
Let $j \in C$. The final solution will open a facility $i$ such that $d(i,j) \leq 4L + 6\epsilon L$.
\end{lemma}

\begin{proof}
The proof will follow the reasoning of that of \Cref{lm:rad-int-uni}, but utilizing the bounds from \Cref{lm:rad-half-uni} in place of those of \Cref{lm:rad-half-uni}.  

Additionally, we must use the fact that points selected in Stage 3 that make up $C'$ (i.e. the set of points whose primary or secondary facility will be opened) will contain all points from $C_L$, i.e. for all points $j \in C/C'$, $j \in C_T$. Assume this is not the case, i.e. that there exists some $j' \in C/C'$ such that $j' \in C_L$. If this occurs, then there will be some $\ell$ such that $S_\ell$ and $S_{j'}$ intersect, and $\ell$ gets selected in the clustering procedure first. Since $j' \in C_L$ and therefore $j' \in C_s$, we know that $\sigma(j') = j'$. Hence, the only way for $S_j$ and $S_\ell$ to intersect is if $\sigma(\ell) = j'$. However, this will never occur per the proof of \Cref{lm:rad-half-uni-imp}, in which we showed that for all $j \in C_b$, $\sigma(j) \in C_T$. Now, $\ell$ must be in $C_b$ if $\sigma(\ell) \neq \ell$, indicating a contradiction. Therefore, we know that $\forall j \in C/C'$, $j \in C_T$, and hence the bound from the first part of \Cref{lm:rad-half-uni-imp}(a) should be utilized. 
\end{proof}

Finally, the following Theorem can be obtained by using \Cref{lm:rad-filter-uni-imp2,lm:rad-int-uni-imp}. 

\begin{theorem} \label{thm:rad-uni-imp}
Let $S$ be the output of the aforementioned approximation algorithm for \upmatmed. For all $k \in \mC$, $d(k,S) \leq 5L + 8\epsilon L$. 
\end{theorem}

\begin{proof}
Let $k \in D(j)$ for $j \in C$, after filtering. If $j \in C_L$, then $k \in \mC_L$. Using previous results, we have $d(k,S) \leq d(j,k) + d(k, S) \leq 2L + L$.

If $j \in \mC_T$, then $k$ can either be from $\mC_T$ or $\mC_L$. If $k \in \mC_T$, then we will have  $d(k,S) \leq d(j,k) + d(k,S) \leq 4\epsilon L + 4L + 6\epsilon L = 4L + 10 \epsilon L$

If $k \in \mC_L$, then we have $d(k,S) \leq d(j,k) + d(k,S) \leq (1+2\epsilon)L + 4L + 6\epsilon L = 5L + 8 \epsilon L$.
\end{proof}

\begin{remark} Without changing \Cref{ln:chld}, using only the updated setting for $\lambda$, we would still be able to achieve a $(6 + 6\epsilon, 4 + 2/\epsilon)$-approximate solution in this setting. Note that the change made to \Cref{ln:chld} can be done even in the non-uniform case. For $\pmatmed$, this change would have resulted in an increase to the radius guarantee. For the general $\upmatmed$ case (\Cref{sec:upmatmed}), this change would have neither worsened nor improved the radius guarantee. 
\end{remark}

\bibliography{references}

\end{document}